\documentclass[journal]{IEEEtran}
\ifCLASSINFOpdf
  \usepackage[pdftex]{graphicx}
\else
\fi

\usepackage{amsmath,amsfonts,amsthm}
\usepackage{color}
\usepackage{enumerate, comment}
\usepackage{bm}
\usepackage{cite}
\usepackage{booktabs}
\usepackage{multirow}
\usepackage{caption}

 \newtheorem{theorem}{Theorem}
  \newtheorem{remark}{Remark}  
    \newtheorem{definition}{Definition}  
          
    \newtheorem{corollary}{Corollary}

\begin{document}

\title{A Chance-Constrained Stochastic Electricity Market}

\author{Yury Dvorkin, \textit{IEEE, Member}\thanks{Yury Dvorkin is with the Department of Electrical and Computer
Engineering, Tandon School of Engineering, New York University, New York,
NY 11201 USA (e-mail: dvorkin@nyu.edu). This work was supported by the U.S. National Science Foundation under
Grants \# ECCS-1760540, CMMI-1825212, ECCS-1847285 and by the Alfred P. Sloan Foundation under Grant \#  G-2019-12363.}}

\maketitle

\begin{abstract}
    Efficiently accommodating uncertain renewable resources in wholesale electricity markets is among the foremost priorities of market regulators in the US, UK and EU nations. However, existing deterministic market designs fail to internalize the uncertainty and  their scenario-based stochastic extensions are limited in their ability to simultaneously maximize social welfare and guarantee non-confiscatory market outcomes in expectation and per each scenario. This paper proposes a chance-constrained stochastic market design, which is capable of producing a robust competitive equilibrium and internalizing uncertainty of the renewable  resources in the price formation process. The equilibrium and resulting prices are obtained for different uncertainty assumptions, which requires using either linear (restrictive assumptions) or second-order conic (more general assumptions) duality in the price formation process. The usefulness of the proposed stochastic market design is demonstrated via the case study carried out on the 8-zone ISO New England testbed.

\end{abstract}

\IEEEpeerreviewmaketitle

\section{Introduction}

Following the restructuring of the power  sector in the US and many European nations, wholesale electricity markets have become instrumental for unleashing  competitive forces that, at least theoretically, should encourage  efficiency improvements among electricity suppliers and eventually reduce the cost of electricity for consumers. For example, PJM  reports that their market has annually saved up to \$2.3 billion and reduced wholesale electricity prices by 40\% in 2008-2017 \cite{pjm_value_of_the_market_2018}. 
However, there is a growing concern that the ability of existing electricity market designs 
to continue delivering these benefits will drastically diminish, as the current trend to massively deploy large-scale renewable resources continues. This concern is mainly attributed to the uncertainty and limited controllability of  renewable resources, as well as their zero or near-zero production costs, which tend to distort market outcomes by dispatching thermal generators in an out-of-merit order  \cite{CONEJO2018520, Bose2019}. \textcolor{black}{To account for the effect of constantly increasing uncertainty on market outcomes, Morales et al. \cite{6656976} redefine the merit order to include the expected cost of uncertainty, in addition to the original marginal cost of production, associated with increased reserve needs due to the presence of renewable resources.} Therefore, consistent with the motivation in \cite{6656976,CONEJO2018520, Bose2019}, the US Department of Energy  emphasizes the need for  `\textit{market structures such as ancillary services, balancing markets and energy markets [that] maintain their competitive frameworks [...] as resource mixes change to assure that the market rules are providing the appropriate signals to bring forth both long-term and short-term electricity supplies}' \cite{osti_1251315}.

This need has paved the way for new market mechanisms, commonly referred to as stochastic market designs, that are capable of holistically modeling probabilistic characteristics of renewable resources, e.g., by means of scenario-based stochastic programming. To a large extent, these mechanisms are enabled by the seminal work of Papavasiliou and Oren \cite{5739566}, which demonstrated the economic savings of scenario-based stochastic programming attained from reducing overly conservative  deterministic reserve margins \cite{5739566}.  Thus, instead of using conservative, exogenously set margins (e.g., $(3+5)$-rule as in \cite{5739566}), uncertain and variable outputs of renewable resources can be represented via a finite  set of  scenarios and their
corresponding probabilities. This leads to a lower expected and ex-post operating cost and, under the assumption of inflexible demand, maximizes the social welfare. While this welfare-maximization is a desired property of any market design, existing stochastic market designs struggle to achieve it simultaneously with two other desired properties -- revenue adequacy, i.e., the payments collected by the market from consumers are greater or equal to the payments made by the market to  producers, and cost recovery, i.e., the payment to each producer is greater or equal to its operating cost\footnote{Current deterministic US markets also use out-of-market corrections and uplift payments to retain market participants.}.  Furthermore, in the specific case of scenario-based stochastic programming, achieving revenue adequacy and cost recovery is difficult since it must be done for both the expected case and each scenario individually. For instance, Pritchard et al. \cite{doi:10.1287/opre.1090.0800}, Morales et al. \cite{6146389} and Wong et al. \cite{4162624} demonstrated that  revenue adequacy and cost recovery are satisfied in expectation, but do not necessarily hold for individual scenarios. Kazempour et al. \cite{8246578} and Ruiz et al. \cite{6153412} simulated a stochastic electricity market using stochastic equilibrium problems to simultaneously ensure cost recovery and revenue adequacy per scenario and in expectation. However, the  market designs in \cite{6153412,8246578}  do not guarantee social welfare maximization and, therefore, are intended by the authors for market analyses rather than market-clearing tools. As discussed in \cite{8246578},  a lack of cost recovery and revenue adequacy guarantees inhibits  implementing scenario-based stochastic market designs in practice. As a result, rare real-world  stochastic electricity markets are limited to exogenous  sizing of probabilistic security margins (reserves) within otherwise deterministic market-clearing routines, e.g., in  Swissgrid \cite{7081775}.

Realizing the shortcomings of scenario-based stochastic programming described above, this paper offers a different perspective on a stochastic market design. 
Instead of modeling uncertain outputs of renewable resources by means of a set of scenarios as in \cite{doi:10.1287/opre.1090.0800,6146389,4162624, MORALES2014765, 8246578, 6153412, 7081775}, one can exploit a chance-constrained approach to internalize the stochasticity of renewable resources in market-clearing tools using  statistical moments of the underlying uncertainty (e.g., mean and standard deviation). This approach leads to  chance (probabilistic) constraints that, in turn, can be exactly reformulated into convex,
deterministic expressions and solved efficiently at scale \cite{7332992, doi:10.1137/130910312}. Furthermore,  these chance constraints  offer a high degree of modeling fidelity to control uncertainty assumptions (e.g., probability distributions \cite{6652224, 7268773}) and risk tolerance (e.g., the likelihood of constraint violations \cite{6652224, doi:10.1137/130910312, 7332992, 8600344}). Replacing a set of scenarios with its statistical moments  using chance constraints not only offers a more accurate representation of uncertainty in market-clearing and dispatch tools, see comparison in \cite{6652224, 7268773, doi:10.1137/130910312}, but also eliminates the need to trade-off between expected and per scenario performance, while immunizing the resulting market outcomes against uncertainty. That is, a stochastic  solution is obtained at the expense of solving a deterministic optimization problem, which internalizes statistical moments and risk parameters in the price formation process.

 This paper  proposes an alternative stochastic market design that uses chance constraints  to accurately  model uncertainty of renewable resources. Relative to scenarios, which are often difficult to obtain, the chance constraints can be formulated using  statistical moments of uncertain quantities, which are readily available from historical observations\footnote{\textcolor{black}{Such historical observations can either be collected by the market operator from their day-to-day operations, or obtained from public repositories supported by the US National Aeronautics and Space Administration and  National Oceanic and Atmospheric Administration \cite{7268773}, or licensed/purchased from third-party data providers \cite{open_power}.}} \cite{7268773},  and internalize this uncertainty in market-clearing tools. To this end, we formulate a two-stage chance-constrained unit commitment (CCUC) problem that follows pool-market assumptions typical for US wholesale electricity markets.  Within this CCUC problem, we consider three different  assumptions on underlying uncertainty. First, we assume that the uncertainty is represented by a normal distribution. In this case, the CCUC problem is reduced to a mixed-integer linear  program (MILP) that can be used for pricing electricity similarly to the current US practice (e.g., as in \cite{ONEILL2005269,8329412}). Second, to better accommodate realistic uncertainty, which is often not normally distributed \cite{6942382}, and quadratic production costs of thermal generators, we formulate distributionally robust chance constraints and approximate them using the Chebyshev approximation, leading to a mixed-integer second-order conic (MISOC) program. Third, since the Chebyshev approximation is notoriously conservative, we invoke an exact second-order conic  (SOC) reformulation of distributionally robust chance constraints from \cite{7973099}, which also renders a MISOC program. \textcolor{black}{Note that the second and third assumptions cannot be accommodated for electricity pricing by means of linear duality theory, as in \cite{ONEILL2005269,8329412}, and thus the resulting MISOC programs require using more general SOC duality.} This paper proves that the MISOC equivalents of the CCUC problem yield a robust competitive equilibrium and analyzes electricity prices obtained by means of SOC duality. In addition to its superior computational performance relative to scenario-based stochastic programming \cite{doi:10.1137/130910312}, using the CCUC for electricity pricing is advantageous in several aspects. First, the market-clearing procedure does not rely on scenarios and produces a single set of market decisions. Hence,  market participants, who are currently   distrustful of a scenario-based stochastic market with scenario parameters they do not control \cite{8246578}, will not be exposed to risk of losses. Second, the prices obtained from the proposed market internalize  uncertainty and risk parameters in the price formation process, without trading off between expected and per scenario performance. As a result, the proposed chance-constrained market design enables  real-world implementations of stochastic electricity markets.

\section{Stochastic Market via Chance Constraints}

Following the current US practice, we formulate a two-stage CCUC problem that optimizes the power production for a single time instance in the future, where  the only source of uncertainty stems from  wind power generation:
\allowdisplaybreaks
\begin{subequations}
\label{eq:ccuc}
\begin{flalign}
    \min_{p_i (\bm{\omega}), p_i, \alpha_i, u_i}	&\mathbb{E}_{\bm{\omega}} \sum_{i \in \mathcal{I}} \bigg[C_{0,i} u_{i} + C_{1,i} p_i (\bm\omega) +   C^{}_{2,i} p_i^2 (\bm\omega) \bigg] \label{eq:ccuc_obj}	\\
    & \hspace{-3mm}  p_{i} (\bm{\omega}) = p_i - \alpha_i \bm{\omega}, \quad \forall i \in \mathcal{I} \label{eq:ccuc_response} \\
    & \hspace{-3mm}  \mathbb{P}_{\omega} \big[u_i \overline{P}_i  \leq p_{i} (\bm{\omega})  \leq \underline{P}_i u_i, \forall i \in \mathcal{I}  \big] \geq 1 -\textcolor{black}{\epsilon} \label{eq:ccuc_jointcc} \\  
    & \hspace{-3mm}\textcolor{black}{\alpha_i \leq u_{i}, \quad \forall i \in \mathcal I }  \label{eq:ccuc_feasibility_alpha} \\
    & \hspace{-3mm}  \sum_{i \in \mathcal{I}} p_i (\bm{\omega}) + W + \bm{\omega} = D \label{eq:ccuc_power_balance} \\ 
	& \hspace{-3mm} \sum_{i \in \mathcal{I}} \alpha_i = 1  \label{eq:ccuc_integrality}  \\
	& \hspace{-3mm} p_{i} \geq 0, 0 \leq \alpha_i \leq 1, u_i \in \{ 0, 1  \}, \quad \forall i \in \mathcal{I},    \label{eq:ccuc_variables}
\end{flalign}
\end{subequations}
\allowdisplaybreaks[0] where $u_i$  is a binary decision on the on/off status of controllable generator $i$  from set $\mathcal I$ and  $p_i (\bm{\omega})$ is the power output of this generator under uncertainty $\bm{\omega}$.  \textcolor{black}{Similarly to the current practice, in which the market operator provides forecasts and estimates the reserve requirements, it is assumed that forecasts $D$ and $W$, as well as parameters characterizing uncertainty $\bm{\omega}$ (e.g., distribution and statistical moments), are released by the market operator and are common for all market participants. Assuming that common forecasts are shared among market participants makes it possible to neglect the effects of information asymmetry that can be exploited by strategically acting market participants with proprietary (and different from the market)  information, see \cite{DVORKIN2019521}.} 
Eq.~\eqref{eq:ccuc_obj} minimizes the expected operating cost given decisions  $u_i$ and  $p_i (\bm{\omega})$ and production cost of each controllable generator given by coefficients $C_{2,i}$, $C_{1,i}$ and $C_{0,i}$. The output of generator $i$ under uncertainty is modeled   using a proportional control law in \eqref{eq:ccuc_response}, where $p_i $ is  a scheduled power output and  $\alpha_i$ is a reserve participation factor. \textcolor{black}{Note that the control law in \eqref{eq:ccuc_response} assumes that recourse decision $p_i (\bm{\omega})$ is parameterized in terms of first-stage decisions $p_i $ and $\alpha_i$, which corresponds to preventive security when the  operator aims to withstand uncertainty realizations without corrective control actions. Alternatively, one can replace \eqref{eq:ccuc_response} with a corrective recourse as explained in \cite{8017474}.}  The \textit{joint, two-sided} chance constraint in \eqref{eq:ccuc_jointcc} ensures that  $p_i (\bm{\omega})$  is within the minimum ($\underline{P}_i$) and maximum ($\overline{P}_i$) power output limits with the probability given by \textcolor{black}{$(1-\epsilon)$}, where \textcolor{black}{$\epsilon>0$} is a small number that represents the tolerance of the market to constraint violations. We assume that wind producers are modeled as undispatchable price-takers with the uncertain power outputs of $W + \bm{\omega}$, where  $W$ is a given forecast and  $\bm{\omega}$ is its uncertainty.  \textcolor{black}{Eq.~\eqref{eq:ccuc_feasibility_alpha} ensures that participation factor $\alpha_i=0$, if controllable generator $i$ is offline, i.e., $u_i =0$, and  attains a non-negative value from its domain range $0\leq \alpha_i \leq 1$, if otherwise.}   The system-wide power balance is enforced in \eqref{eq:ccuc_power_balance}, which balances the total output of conventional and wind power generation resources and demand. 
Eq.~\eqref{eq:ccuc_integrality} ensures  the sufficiency of reserve provided by controllable generators to cope with  uncertainty $\bm{\omega}$. The decision variables are declared in \eqref{eq:ccuc_variables}. Solving the CCUC in  \eqref{eq:ccuc} depends on  the treatment of \eqref{eq:ccuc_jointcc} and the assumptions made on  $\bm{\omega}$ as discussed below. 
\subsubsection{Approximation by individual chance constraints} To avoid dealing with the joint, two-sided chance constraint in \eqref{eq:ccuc_jointcc}, it is common to invoke two ad-hoc assumptions that follow from power system practices. First, it is assumed that violations on different generators are independent of one another during normal (steady-state) power system operations. Second, simultaneous violations of the minimum and maximum output limits on a given conventional generator are impossible. As a result,  \eqref{eq:ccuc_jointcc} can be  approximated by the following separate, one-sided chance constraints:
\allowdisplaybreaks
\begin{subequations}
\label{eq:ccuc_single_cc}
\begin{flalign}
    & \mathbb{P}_{\omega} \big[ p_{i} (\bm{\omega})  \leq \overline{P}_i u_i \big] \geq 1 -\epsilon_i,  \quad \forall i \in \mathcal{I} \label{eq:ccuc_single_cc_max} \\  
    & \mathbb{P}_{\omega} \big[u_i \underline{P}_i  \leq p_{i} (\bm{\omega}) \big]  \geq 1 -\epsilon_i,  \quad \forall i \in \mathcal{I}, \label{eq:ccuc_single_cc_min}  
\end{flalign}
\allowdisplaybreaks[0]
\textcolor{black}{where parameter $\epsilon_i$ denotes the tolerance of the market to constraint violations at controllable generator $i$. The value of $\epsilon_i$ is typically set to a small positive number, which is anticipated to vary for different systems based on the market's security preferences, and can be chosen via extensive data-driven simulations, e.g., as in \cite{6714594}, to meet a given ad-hoc reliability criteria (e.g., expected energy not served or loss of probability metrics). Previous  studies, e.g.,  \cite{8017474, 8600344, 7332992}, have also shown that this approximation limits the joint violation probability effectively due to a  few simultaneously active constraints. Further treatment of \eqref{eq:ccuc_single_cc} depends on the assumption made on uncertainty $\bm{\omega}$.} \textcolor{black}{Data-driven analyses in  \cite{7268773} show that  $\bm{\omega}$ can be parameterized using $\mu = \mathbb{E} [{\bm{\omega}}]$ as the mean forecast error and $\sigma^2 = Var[{\bm{\omega}}]$ as its variance. Given $\mu = \mathbb{E} (\bm{\omega}) $  and $\sigma^2 = Var[{\bm{\omega}}]$, terms $\mathbb{E}_{\bm{\omega}} \big[ C_{1,i} p_i (\bm\omega) \big] $ and  $\mathbb{E}_{\bm{\omega}}\big[ C_{2,i} p_i^2 (\bm\omega)\big]$ in \eqref{eq:ccuc_obj} can be replaced with the following deterministic expressions:}

\begin{flalign} 
&  \textcolor{black}{\mathbb{E}_{\bm{\omega}} \big[ C_{1,i} p_i (\bm\omega)\big] \!=\! C_{1,i}  \mathbb{E}_{\bm{\omega}} \big[( p_i -\alpha_i \bm\omega) \big]\!\!=\!\! C_{1,i}  ( p_i -\alpha_i \mu)} \label{eq:ccuc_single_cc1a} \\
& \textcolor{black}{\mathbb{E}_{\bm{\omega}} \big[ C_{2,i} p_i^2 (\bm\omega)\big] \!=\! C_{2,i}  \mathbb{E}_{\bm{\omega}} \big[( p_i -\alpha_i \bm\omega)^2 \big]\!\!=\!\! \mathbb{E}_{\bm{\omega}} \big[ C_{2,i}  ( p_i^2 -  \nonumber} \\
& \textcolor{black}{2  \bm\omega \alpha_i p_i + (\bm\omega \alpha_i)^2 \big] = C_{2,i}   p_i^2 -2  \mu \alpha_i p_i + \alpha_i^2 (\sigma^2 + \mu^2) }. \label{eq:ccuc_single_cc1b}
\end{flalign}
\end{subequations}
\textcolor{black}{Note that \eqref{eq:ccuc_single_cc1a}-\eqref{eq:ccuc_single_cc1b} are derived only using $\mu = \mathbb{E} (\bm{\omega}) $  and $\sigma^2 = Var[{\bm{\omega}}]$ and do not assume a particular parametric distribution (e.g., normal). Thus, using the result in \eqref{eq:ccuc_single_cc1a}-\eqref{eq:ccuc_single_cc1b} in  \eqref{eq:ccuc_obj} and  invoking that $\bm{\omega}\sim N(\mu, \sigma^2)$ to reformulate \eqref{eq:ccuc_single_cc_max}-\eqref{eq:ccuc_single_cc_min},  the CCUC problem in \eqref{eq:ccuc} is recast as:}
\allowdisplaybreaks
\begin{subequations}
\label{eq:ccuc_single_reformulation}
\begin{flalign}
    & \min_{p_i, \alpha_i, u_i} \sum_{i \in \mathcal{I}} \bigg[C_{0,i} u_{i} +  C_{1,i} (p_i - \mu \alpha_i) + \nonumber \\ &   \quad \quad \quad \quad \quad   C^{}_{2,i} \big(p_i^2 - 2 \mu \alpha_i p_i + \alpha_i^2 (\sigma^2 + \mu^2) \textcolor{black}{\big)} \bigg] \label{eq:ccuc_single_reformulation_obj}  \\
    &  p_{i} \leq \overline{P}_i u_i - \hat \sigma_i \alpha_i,  \quad \forall i \in \mathcal{I} \label{eq:ccuc_single_reformulation_max} \\   
    &   - p_{i} \leq -\underline{P}_i u_i - \hat \sigma_i \alpha_i ,  \quad \forall i \in \mathcal{I} \label{eq:ccuc_single_reformulation_min} \\   
        & \textcolor{black}{\alpha_i \leq u_{i}, \quad \forall i \in \mathcal I }  \label{eq:ccuc_single_reformulation_alpha} \\
    & \sum_{i \in \mathcal{I}} p_i  = D  - W \label{eq:ccuc_single_reformulation_power_balance}  \\   
    & \sum_{i \in \mathcal{I}} \alpha_i= 1   \label{eq:ccuc_single_reformulation_integrality} \\     
	&  p_{i} \geq 0, \alpha_i \geq 0, u_i \in \{ 0, 1  \}, \quad \forall i \in \mathcal{I},   \label{eq:ccuc_single_reformulation_variable}  
\end{flalign}
\end{subequations}  
\allowdisplaybreaks[0]
where  $\hat \sigma_i = ( \Phi^{-1}(1-\epsilon_i) \sigma -\mu)$ is a given parameter and $\Phi^{-1}(\cdot)$ is the  inverse cumulative distribution function of the standard normal distribution. \textcolor{black}{Note that, if $u_i=0$ in  \eqref{eq:ccuc_single_reformulation}, it follows from  \eqref{eq:ccuc_single_reformulation_alpha} that $\alpha_i = 0$ and $p_i=0$ due to \eqref{eq:ccuc_single_reformulation_max}-\eqref{eq:ccuc_single_reformulation_min}.} While  the constraints of \eqref{eq:ccuc_single_reformulation} are linear, the  objective function in \eqref{eq:ccuc_single_reformulation_obj} is quadratic, thus turning \eqref{eq:ccuc_single_reformulation} into a mixed-integer quadratic  program (MIQP), which can be solved by off-the-shelf solvers (e.g., \textcolor{black}{CPLEX}, Gurobi).

\subsubsection{Approximation by the Chebyshev inequality}
While the normal assumption on $\bm{\omega}$ in \eqref{eq:ccuc_single_reformulation} fares well in practice, e.g.,
\cite{doi:10.1137/130910312, 7332992}, it introduces some inaccuracies as empirically measured  uncertainty does not follow this distribution exactly, e.g., \cite{6942382, 7268773}. To overcome this limitation, $\bm{\omega}$ can be modeled using a set of distributions, rather than a single  distribution as in \eqref{eq:ccuc_single_reformulation}:
\begin{flalign} 
\Omega = \big\{ \mathbb{P}: \mathbb{E}_{\mathbb{P}} [\bm{\omega}]=\mu, \mathbb{E}_{\mathbb{P}} [\bm{\omega}^2]=\sigma^2\big\}, \label{eq:uncertainty_set}
\end{flalign} 
where uncertainty set $\Omega$ encapsulates all probability measures with given first- and second-order   moments $\mathbb{E}_{\mathbb{P}} [\bm{\omega}]$ and $\mathbb{E}_{\mathbb{P}} [\bm{\omega}^2]$. Assuming $\bm{\omega}\! \in\! \Omega$ makes it possible to recast  \eqref{eq:ccuc_jointcc} as  distributionally robust chance constraints \cite{XieThesis,7973099, SUMMERS2015116}:
\allowdisplaybreaks
\begin{subequations}
\label{eq:single_inf}
\begin{flalign}  
    & \inf_{\mathbb{P}_\omega \in \Omega } \mathbb{P}_{\omega} \big[ p_{i} (\bm{\omega})  \leq \overline{P}_i u_i \big] \geq 1 -\epsilon_i,  \quad \forall i \in \mathcal{I} \label{eq:ccuc_single_inf_max} \\  
    & \inf_{\mathbb{P}_\omega \in \Omega } \mathbb{P}_{\omega} \big[u_i \underline{P}_i  \leq p_{i} (\bm{\omega}) \big]  \geq 1 -\epsilon_i,  \quad \forall i \in \mathcal{I}. \label{eq:ccuc_single_inf_min}
\end{flalign}
\end{subequations}
\allowdisplaybreaks[0]
Applying the Chebyshev inequality to \eqref{eq:single_inf} as described in \cite{XieThesis,7973099, SUMMERS2015116},  the  CCUC problem in \eqref{eq:ccuc} can be replaced with:
\allowdisplaybreaks
\begin{subequations}
\label{eq:ccuc_chebyshev}
\begin{flalign} 
    & \min_{p_i, \alpha_i, u_i} \sum_{i \in \mathcal{I}} \bigg[C_{0,i} u_{i} +  C_{1,i} (p_i - \mu \alpha_i) + \nonumber \\ &   \quad \quad \quad \quad \quad   C^{}_{2,i} \big( p_i^2 - 2 \mu \alpha_i p_i + \alpha_i^2 (\sigma^2 + \mu^2) \textcolor{black}{\big)} \bigg]  \label{eq:ccuc_chebyshev_1} \\
    & \hspace{10mm} p_i \leq \overline{P}_i u_i - \tilde \sigma_i \alpha_i, \quad \forall i \in \mathcal{I} \\
    & \hspace{10mm} - p_i \leq - \underline{P}_i u_i - \tilde \sigma_i \alpha_i, \quad \forall i \in \mathcal{I} \\    
    & \hspace{10mm} \textcolor{black}{\alpha_i \leq u_{i}, \quad \forall i \in \mathcal I }  \label{eq:ccuc_chebyshev_alpha} \\
    & \hspace{10mm}\sum_{i \in \mathcal{I}} p_i  = D  - W \label{eq:ccuc_single_reformulation_power_balance2}  \\   
    & \hspace{10mm} \sum_{i \in \mathcal{I}} \alpha_i= 1   \label{eq:ccuc_single_reformulation_integrality2} \\     
	&  \hspace{10mm} p_{i} \geq 0, \alpha_i \geq 0, u_i \in \{ 0, 1  \}, \quad \forall i \in \mathcal{I},   \label{eq:ccuc_single_reformulation_variable2}  
\end{flalign}
\end{subequations}
\allowdisplaybreaks[0]
where $\tilde \sigma_i = (\sqrt{\frac{1-\epsilon_i}{\epsilon_i}} \sigma - \mu)$. Similarly to \eqref{eq:ccuc_single_reformulation},  \eqref{eq:ccuc_chebyshev} is a MIQP that can be solved efficiently with off-the-shelf solvers. Although adjusting parameter $\tilde\sigma_i$  allows for better fitting  of empirical data on uncertainty $\bm{\omega}$,  the accuracy of the Chebyshev approximation reduces when $\epsilon_i \rightarrow 0$ and its solution becomes unnecessarily conservative or  may even be infeasible,  \cite{7973099}. Note that $\epsilon_i$ in \eqref{eq:ccuc_chebyshev}  can be chosen such that  $\tilde\sigma_i=\hat\sigma_i$, i.e.,  \eqref{eq:ccuc_single_reformulation} and \eqref{eq:ccuc_chebyshev} yield identical solutions.  

\subsubsection{Exact SOC reformulation}

Motivated by the need to overcome the conservatism of the Chebyshev approximation in \eqref{eq:ccuc_chebyshev}, Xie and Ahmed  \cite{7973099} derived an  SOC equivalent of  \eqref{eq:ccuc_jointcc}. Using  \cite[Theorem 2]{7973099},  the CCUC problem in \eqref{eq:ccuc} is equivalent to:
\allowdisplaybreaks
\begin{subequations}
\label{eq:ccuc_drcc}
\begin{flalign} 
    & \min_{p_i, \alpha_i, u_i, y_i, \pi_i} \sum_{i \in \mathcal{I}} \bigg[C_{0,i} u_{i} +  C_{1,i} (p_i - \mu \alpha_i) + \nonumber  \\ &   \quad \quad \quad \quad \quad   C^{}_{2,i} \big( p_i^2 - 2 \mu \alpha_i p_i + \alpha_i^2 (\sigma^2 + \mu^2) \textcolor{black}{\big)} \bigg] \label{{eq:ccuc_drcc_1}}\\
    & p_i - y_i - \pi_i \leq  \frac{\overline{P}_i + \underline{P}_{\textcolor{black}{i}}}{2} u_i , \quad \forall i \in \mathcal{I} \label{eq:ccuc_drcc_up} \\
    & -p_i - y_i - \pi_i \leq - \frac{\overline{P}_i + \underline{P}_{\textcolor{black}{i}}}{2} u_i, \quad \forall i \in \mathcal{I} \label{eq:ccuc_drcc_down} \\
    & y_i^2 + \alpha_i^2 \sigma^2\leq \epsilon_i \bigg( \frac{\overline{P}_i - \underline{P}_i}{2} - \pi_i \bigg)^2, \quad \forall i \in \mathcal{I}  \label{eq:ccuc_drcc_socp} \\
    & \textcolor{black}{\alpha_i \leq u_{i}, \quad \forall i \in \mathcal I }  \label{eq:ccuc_drcc_alpha} \\    
    & \sum_{i \in \mathcal{I}} p_i = D  -W \label{eq:ccuc_drcc_power_balance}  \\   
    & \sum_{i \in \mathcal{I}} \alpha_i= 1    \\     
	&  p_{i}\! \geq\! 0, 0\! \leq\! \alpha_i \!\leq\! 1, y_{i}\! \geq\! 0, 0\!\leq\! \pi_{i}\! \leq \!\frac{\overline{P}_i\!-\!\underline{P}_i}{2}\!, u_i \!\in\! \{ 0, 1  \}, \forall i \in \mathcal{I},   
\end{flalign}
\end{subequations}
\allowdisplaybreaks[0]
where $y_i$ and $\pi_i$ are auxiliary variables and \eqref{eq:ccuc_drcc_up}-\eqref{eq:ccuc_drcc_socp} are exact equivalents of \eqref{eq:ccuc_jointcc}. Relative to the optimization in \eqref{eq:ccuc_single_reformulation} and \eqref{eq:ccuc_chebyshev} that only have linear constraints, the notable difference of \eqref{eq:ccuc_drcc} is constraint \eqref{eq:ccuc_drcc_socp}, which is a SOC constraint.

\subsection{Pricing with Chance Constraints via LP duality} \label{sec:lp_duality}

\subsubsection{\textcolor{black}{Prior work}}Motivated by the current practice of electricity markets to use LP duality for obtaining electricity prices, \cite{8329412} proposed to reduce the MIQP  in \eqref{eq:ccuc_single_reformulation} to a linear program (LP) by invoking two restrictive assumptions that  $\mu =0$, i.e.,  $\bm{\omega}\sim N(0,\sigma^2)$, and  $C_{2,i}=0$, which leads to the following MILP:
\allowdisplaybreaks
\begin{subequations}
\label{eq:ccuc_milp_reformulation}
\begin{flalign}
    & \min_{p_i, \alpha_i, u_i} \sum_{i \in \mathcal{I}} \bigg[C_{0,i} u_{i} +  C_{1,i} p_i \bigg]   \\
    & p_{i} \leq \overline{P}_i u_i - \hat \sigma_i \alpha_i,  \quad \forall i \in \mathcal{I} \label{eq:ccuc_drcc_up2} \\   
    & - p_{i}  \leq - \underline{P}_i u_i - \hat \sigma_i \alpha_i,  \quad \forall i \in \mathcal{I} \label{eq:ccuc_drcc_up22} \\       
        & \textcolor{black}{\alpha_i \leq u_{i}, \quad \forall i \in \mathcal I }  \label{eq:ccuc_drcc_alpha2} \\    
   & \sum_{i \in \mathcal{I}} p_i  = D  -W  \\   
  & \sum_{i \in \mathcal{I}} \alpha_i= 1   \\     
	&  p_{i} \geq 0, 0 \leq \alpha_i \leq 1, u_i \in \{ 0, 1  \}, \quad \forall i \in \mathcal{I},   
\end{flalign}
\end{subequations}  
\allowdisplaybreaks[0]
\begin{remark}
 \normalfont  \textcolor{black}{The MILP in \eqref{eq:ccuc_milp_reformulation} can be related to currently used deterministic market-clearing procedures, if the reserve contribution of each controllable generator $i$ is expressed in terms of $\alpha_i$. Indeed, the upward ($r_i^\uparrow$) and downward ($r_i^\downarrow$) reserve margins in \eqref{eq:ccuc_drcc_up2} and \eqref{eq:ccuc_drcc_up22} due to $\alpha_i$ can be computed as $r_i^\uparrow =r_i^\downarrow =\alpha_i \hat \sigma_i$. Accordingly, the total upward and downward reserve requirement allocated in \eqref{eq:ccuc_milp_reformulation} can be computed as $\sum_{i \in \mathcal I} r_i^\uparrow =\sum_{i \in \mathcal I} r_i^\downarrow =  \sum_{i \in \mathcal I} \alpha_i \hat \sigma_i$.}
\end{remark}

While the  MILP in \eqref{eq:ccuc_milp_reformulation} cannot be used for pricing electricity directly due to the presence of binary variables $u_i$, which prevents computing dual variables of binding constraints, it can be converted into an equivalent LP problem, which can be used for electricity pricing as proven in \cite{ONEILL2005269,8329412}.  First, \eqref{eq:ccuc_milp_reformulation} is solved using a MILP solver (e.g., \textcolor{black}{CPLEX}, Gurobi) to obtain the optimal values of binary variables $u^*_i$. Second, the following LP equivalent of \eqref{eq:ccuc_milp_reformulation} is solved to obtain dual variables:
\allowdisplaybreaks
\begin{subequations}
\label{eq:ccuc_lp_reformulation}
\begin{flalign}
    & \min_{p_i, \alpha_i, u_i} \sum_{i \in \mathcal{I}} \bigg[C_{0,i} u_{i} +  C_{1,i} p_i \bigg]   \\
(\overline{\mu}_i ): & \hspace{3mm}   p_{i} \leq \overline{P}_i u_i - \hat \sigma_i \alpha_i,  \quad \forall i \in \mathcal{I} \label{eq:lp_limits} \\
(\underline{\mu}_i): & \hspace{3mm}  - p_{i} \leq -\overline{P}_i u_i - \hat \sigma_i \alpha_i,  \quad \forall i \in \mathcal{I} \label{eq:lp_limits2} \\
    (\lambda): & \hspace{3mm} \sum_{i \in \mathcal{I}} p_i  = D -W    \label{eq:lp_power}  \\   
    (\chi): & \hspace{3mm}  \sum_{i \in \mathcal{I}} \alpha_i= 1  \label{eq:lp_regulation} \\    
     \textcolor{black}{(v_i):}     & \hspace{3mm}  \textcolor{black}{\alpha_i \leq u_{i}, \quad \forall i \in \mathcal I }  \label{eq:lp_alpha2} \\
    (\gamma_i): & \hspace{3mm}   u_i = u^*_i \label{eq:lp_commitment} \\      
	& \hspace{3mm}   p_{i} \geq 0, 0 \leq \alpha_i \leq 1, 0 \leq u_{i} \leq 1, \quad \forall i \in \mathcal{I},   
\end{flalign}
\end{subequations}  
\allowdisplaybreaks[0]
where  variables $u_i$ are converted into real-valued variables and \eqref{eq:lp_commitment} sets the value of this variable  to  $u_i^\ast$.  Since \eqref{eq:ccuc_milp_reformulation} and \eqref{eq:ccuc_lp_reformulation}   yield the same optimal solution, as proven in \cite{ONEILL2005269,8329412}, dual variables  $\lambda$, $\chi$, and $\gamma_i$ of constraints  \eqref{eq:lp_power}-\eqref{eq:lp_commitment} can be leveraged for electricity pricing. Next, \textcolor{black}{\cite{8329412} defines the robust competitive equilibrium as follows:}

\begin{definition} \label{definition1a} \normalfont
\textcolor{black}{A robust competitive equilibrium for the stochastic market defined by \eqref{eq:ccuc_milp_reformulation} is a set of prices $\{\lambda, \chi, \{ \gamma_i, \forall i \in \mathcal{I}\}\}$ and a set of dispatch decisions  $\{ p_i, \alpha_i, u_i, \forall i \in \mathcal{I}\}$ that (i)  clear the market, i.e., $\sum_{i \in \mathcal{I}} p_i   = D-W$ and $\sum_{i \in \mathcal I} \alpha_i = 1$, and (ii) maximize the profit of individual generators.} 
\end{definition} \noindent \textcolor{black}{We now prove that \eqref{eq:ccuc_milp_reformulation} and  \eqref{eq:ccuc_lp_reformulation} return this equilibrium in the following theorem:}

\allowdisplaybreaks
\begin{theorem}\label{theorem1}  \normalfont  Let  $\{p^\ast_i, \alpha^\ast_i, u^\ast_i, \forall i \in \mathcal{I}\}$ be an optimal solution of \eqref{eq:ccuc_milp_reformulation} and let   $\{\lambda^\ast, \chi^\ast, \{\gamma^\ast_i, \forall i \in \mathcal{I}\}\}$ be dual variables of \eqref{eq:ccuc_lp_reformulation}. Then  $\{\{p^\ast_i, \alpha^\ast_i, u^\ast_i, \gamma^\ast_i,\forall i \in \mathcal{I}\}, \lambda^\ast, \chi^\ast\}$ constitutes a robust competitive equilibrium, i.e.:
\begin{enumerate}
    \item The market clears at  $\sum_{i \in \mathcal{I}} \!p^\ast_i \!= \!D- W$ and $ \sum_{i \in \mathcal{I}} \alpha^\ast_i= 1 $.
    \item Each producer maximizes its profit under the payment of $\Gamma_i^\ast = \lambda^\ast  p_i^\ast + \chi^\ast  \alpha_i^\ast +  \gamma_i^\ast u_i^\ast$.
\end{enumerate}
\end{theorem}
\allowdisplaybreaks[0]
\begin{proof} See our previous work in \cite{8329412}.
\end{proof}

In other words, Theorem~\ref{theorem1} establishes that dual variables $\lambda$, $\chi$, and $\gamma$ represent prices for energy, reserve, and commitment allocations that attain the least-cost solution and  support a market equilibrium, i.e., no generator has any incentive to deviate from the solution of \eqref{eq:ccuc_milp_reformulation}. Similarly to the current electricity markets, Theorem~\ref{theorem1} entitles every  generator to receive the following three payments:  (i) $\lambda^\ast  p_i^\ast$ for the energy produced, (ii) $\chi^\ast  \alpha_i^\ast$  for the reserve provided, and (iii) $\gamma_i^\ast  u_i^\ast$ for the commitment status.

\subsubsection{\textcolor{black}{Extensions of prior work}} While still using the LP duality as in \cite{ONEILL2005269}, we extend the results from \cite{8329412} by demonstrating  that prices $\lambda$, $\chi$, and $\gamma_i$   internalize both uncertainty  ($\mu, \sigma$) and risk ($\epsilon_i$) parameters. Consider the stationary  conditions of \eqref{eq:ccuc_lp_reformulation}:
\allowdisplaybreaks
\begin{subequations}
\label{eq:ccuc_lp_price_expressions}
\begin{flalign}
     & \frac{\partial \mathcal L}{\partial  p_i} =  C_{1,i}  + \overline{\mu}_i - \underline{\mu}_i  - \lambda  = 0, \quad \forall i \in \mathcal I\\ 
     & \frac{\partial \mathcal L}{\partial \alpha_i} = \overline{\mu}_i  \hat \sigma_i +  \underline{\mu}_i  \hat \sigma_i - \chi + \textcolor{black}{\upsilon_i}  = 0, \quad \forall i \in \mathcal I\\
     & \frac{\partial \mathcal L}{\partial u_i} = C_{0,i} - \overline{\mu}_i \overline{P}_i + \underline{\mu}_i \underline{P}_i - \textcolor{black}{\upsilon_i} - \gamma_i = 0, \quad \forall i \in \mathcal I,
\end{flalign}
\end{subequations}
where $\mathcal L$ denotes the Lagrangian function of \eqref{eq:ccuc_lp_reformulation}:
\begin{flalign}
 \mathcal L = & \sum_{i \in \mathcal{I}} \bigg[C_{0,i} u_{i} +  C_{1,i} p_i   +
\overline{\mu}_i  ( p_{i} - \overline{P}_i u_i + \hat \sigma_i \alpha_i)  + \underline{\mu}_i (- p_{i}   & \nonumber \\  & +  \underline{P}_i u_i + \hat \sigma_i  \alpha_i) + \textcolor{black}{\upsilon_i (\alpha_i - u_i)} + \gamma_i (u_i^\ast - u_i)\bigg]   \nonumber & \\  & +\lambda \big(  D  - W-   \sum_{i \in \mathcal{I}}p_i \big)  +  \chi \big(  1- \sum_{i \in \mathcal{I}}  \alpha_i \big) &
\end{flalign}

Using the stationary  conditions in \eqref{eq:ccuc_lp_price_expressions} and $\hat \sigma_i = (\Phi^{-1} (1-\epsilon_i) \sigma - \mu )$, prices  $\lambda$, $\chi$, and $\gamma$ can be expressed as follows:
\begin{subequations}
\label{eq:ccuc_lp_price_expressions2}
\begin{flalign}
     & \lambda =  C_{1,i}+\overline{\mu}_i - \underline{\mu}_i, \quad \forall i \in \mathcal I \label{eq:ccuc_lp_price_expressions21}\\
     & \chi  = ( \Phi^{-1}(1-\epsilon_i) \sigma -\mu) (\overline{\mu}_i + \underline{\mu}_i) + \textcolor{black}{\upsilon_i}, \quad \forall i \in \mathcal I\\
    & \gamma_i = C_{0,i} - \overline{P}_i \overline{\mu}_i +\underline{P}_i \underline{\mu}_i - \textcolor{black}{\upsilon_i}, \quad \forall i \in \mathcal I.  \label{eq:ccuc_lp_price_expressions23}
\end{flalign}
\end{subequations} 
\allowdisplaybreaks[0]
As per \eqref{eq:ccuc_lp_price_expressions2}, reserve price $\chi$ explicitly depends on  $\mu$, $\sigma$ and $\epsilon_i$, while  energy and commitment prices $\lambda$  and $\gamma_i$ depend on these parameters implicitly via dual variables of inequality constraints $\underline{\mu}_i$ and $\overline{\mu}_i$. Unlike the scenario-based stochastic market designs in \cite{doi:10.1287/opre.1090.0800, 4162624, 6146389, 8246578, MORALES2014765}, the prices in  \eqref{eq:ccuc_lp_price_expressions2} incorporate  uncertainty and risk parameters without the need to consider multiple scenarios and trading off among per scenario and expected performance.  Notably, if inequality constraints in \eqref{eq:lp_limits} are not binding, i.e., $ \underline{\mu}_i = \overline{\mu}_i = 0 $, these prices reduce to  $\lambda =  C_{1,i}$, $ \chi = 0 $   and $\gamma_i = C_{0,i}$ that matches the prices of the deterministic market design implemented based on \cite{ONEILL2005269}.

\begin{remark} \label{remark2}
  \normalfont  \textcolor{black}{The results in Theorem~\ref{theorem1} and in \eqref{eq:ccuc_lp_price_expressions2}  are obtained for the optimization in \eqref{eq:ccuc_milp_reformulation}  under the assumption that $C_{2,i}=0$ and $\mu =0 $. However, the former assumption can be overcome by adjusting the value of $\epsilon_i$. For example, if the actual uncertainty is modeled as $\bm \omega \sim N (\mu, \sigma^2)$ and the desired tolerance to constraint violations is given by $\hat \epsilon_i$, the optimization \eqref{eq:ccuc_milp_reformulation} is still applicable, despite the underlying assumption that $\mu = 0$, if $\epsilon_i$ in \eqref{eq:ccuc_drcc_up2}-\eqref{eq:ccuc_drcc_up22}  is selected such that  $\Phi^{-1} ( 1- \epsilon_i ) \sigma = (\Phi^{-1} ( 1- \hat \epsilon_i )\sigma - \mu)$. }
  
\end{remark}

\begin{remark}
  \normalfont  Since Theorem~\ref{theorem1} is obtained by reducing \eqref{eq:ccuc} to a MILP and proved using the same procedure as in \cite{ONEILL2005269}, this market design inherits the same cost recovery and revenue adequacy properties as the market design in   \cite{ONEILL2005269} (which currently underlies US markets), i.e., it requires an uplift payment to each generator equal to $\gamma_i^\ast u_i^\ast$ to reflect the cost of commitment decisions.
\end{remark}

\textcolor{black}{Using the uplift payment mechanism, we can show that the equilibrium obtained with Theorem~\ref{theorem1} is sufficient to recover  the operating cost of each producer:}
\begin{corollary} \normalfont \textcolor{black}{Theorem~\ref{theorem1} ensures the full cost recovery by each producer, i.e., $\Pi_i^\ast = \Gamma_i^\ast - C_{0,i} u_i^\ast - C_{1,i} p_i^\ast  =0 $, under the robust competitive equilibrium.}
\end{corollary}
\begin{proof} \color{black}
Using \cite{8329412}, we reformulate \eqref{eq:ccuc_lp_reformulation} as the following equilibrium for a given  value of $u_i^\ast$: 
\begin{subequations}
\label{eq:ccuc_lp_price_expressions21}
\begin{flalign}
    \hspace{-10mm} & \bigg\{ \max_{ p_i, \alpha_i, u_i = u^\ast_i} \Pi_i  \label{eq:eq_lp_11} \\
    \hspace{-10mm} & p_i \leq \overline P_i u_i - \hat \sigma_i \alpha_i \label{eq:eq_lp_2} \\
     \hspace{-10mm} & - p_i \leq - \underline P_i u_i - \hat \sigma_i \alpha_i  \label{eq:eq_lp_3} \\
     \hspace{-10mm} & \alpha_i \leq u_i^\ast \bigg\}, \forall i \in \mathcal I  \label{eq:eq_lp_3b} \\     
       \hspace{10mm} \bigg\{ & (\lambda) :\quad\quad  \sum_{i \in \mathcal I} p_i + W = D  \label{eq:eq_lp_4} \\ 
      \hspace{10mm} &(\chi) :   \quad\quad \sum_{i \in \mathcal I} \alpha_i =1  \label{eq:eq_lp_5} \\ 
       \hspace{10mm} & (\gamma_i) : \quad\quad   u_i =  u_i^\ast, \forall i \in \mathcal I \label{eq:eq_lp_6} \bigg\}, 
\end{flalign}
where \eqref{eq:eq_lp_11}-\eqref{eq:eq_lp_3b} is solved by each producer to maximize their profit and \eqref{eq:eq_lp_4}-\eqref{eq:eq_lp_6} is the market-clearing problem. Let $\{\{p_i^\ast, \alpha_i^\ast, u_i^\ast,  \gamma_i^\ast, \forall i \in \mathcal I\}, \lambda^\ast, \chi^\ast \}$ be the robust equilibrium obtained from Theorem~\ref{theorem1}, i.e., it solves \eqref{eq:ccuc_lp_reformulation} and \eqref{eq:ccuc_lp_price_expressions21}. When \eqref{eq:ccuc_lp_price_expressions21} is solved,  the profit of each generator can be computed as  $\Pi_i^\ast = \Gamma_i^\ast - C_{0,i} u_i^\ast - C_{1,i} p_i^\ast   = \lambda^\ast p_i^\ast + \chi^\ast \alpha_i^\ast + \gamma_i^\ast u_i^\ast  - C_{0,i} u_i^\ast - C_{1,i} p_i^\ast$. To show that  $\Pi_i^\ast\geq 0$ at the optimal solution, we invoke the strong duality theorem for \eqref{eq:eq_lp_11}-\eqref{eq:eq_lp_3b}, which yields:
\begin{flalign}
\Pi_i^\ast   =  v_i u_i^\ast,
\end{flalign}
where $v_i \geq 0$ and, hence, $\Pi_i^\ast \geq 0$. 
\end{subequations} 

\end{proof}
 
Theorem~\ref{theorem1} is developed under the assumption that $\mu =0 $ and $C_{2,i}=0$. The effect of the first assumption on the optimal solution can be mitigated by tuning parameters $\epsilon_i$, \textcolor{black}{see Remark~\ref{remark2}} and the discussion in \cite{6652224}. \textcolor{black}{The second assumption does not allow for accurately\footnote{\textcolor{black}{Note that current electricity markets approximate quadratic production costs using piece-wise linear functions}} computing the expected operating cost, which is crucial for the efficiency of any stochastic electricity market design, see Morales et al. \cite{6656976}.}  \textcolor{black}{These  shortcomings cannot be addressed using LP duality as in \cite{ONEILL2005269,8329412} and motivate the main proposition of this paper, that is to invoke SOC duality for electricity pricing, which  allows for a rigorous stochastic  market-clearing procedure under high-fidelity assumptions on the underlying uncertainty and  internalizing this uncertainty in the market-clearing problem}.

\subsection{Pricing with Chance Constraints via SOC duality}\label{sec:pricing:chebyshev}

This sections deals with electricity pricing for distributionally robust formulations based on the Chebyshev approximation in \eqref{eq:ccuc_chebyshev} and the exact SOC reformulation in \eqref{eq:ccuc_drcc} and  assumes $C_{2,i}> 0$, which inhibits invoking  LP duality.

To show that \eqref{eq:ccuc_chebyshev} and \eqref{eq:ccuc_drcc}   can be used for electricity pricing, we will follow the same procedure as in \cite{ONEILL2005269, 8329412}. We show that the original mixed-integer problem in both cases can be converted into a MISOC program and has an augmented and continuous equivalent (i.e., when the binary decisions are fixed to the optimal value). Second, we will prove that the dual variables of the continuous equivalent are electricity prices and support a robust competitive equilibrium \textcolor{black}{defined similarly to Definition~\ref{definition1a}}:
\begin{definition} \label{definition1} \normalfont
A robust competitive equilibrium for the stochastic market defined by either  \eqref{eq:ccuc_chebyshev} or \eqref{eq:ccuc_drcc} is a set of prices $\{\lambda, \chi, \{ \gamma_i, \forall i \in \mathcal{I}\}\}$ and a set of dispatch decisions  $\{ p_i, \alpha_i, u_i, \forall i \in \mathcal{I}\}$ that (i)  clear the market, i.e., $\sum_{i \in \mathcal{I}} p_i   = D-W$ and $\sum_{i \in \mathcal I} \alpha_i = 1$, and (ii) maximize the profit of individual generators. 
\end{definition}

\subsubsection{Pricing under the Chebyshev approximation} 

Given Definition~\ref{definition1}, our hypothesis is that dispatch decisions $\{ p_i, \alpha_i, u_i, \forall i \in \mathcal{I}\}$ will be obtained by solving the mixed-integer optimization in  \eqref{eq:ccuc_chebyshev} and respective prices will be given by the dual solution of the  following augmented equivalent:
\allowdisplaybreaks
\begin{subequations} \label{eq:augumented_chebyshev}
\begin{flalign}
    & \min_{p_i, \alpha_i, u_i} \sum_{i \in \mathcal{I}} \bigg[C_{0,i} u_{i} +  C_{1,i} p_i  +   C^{}_{2,i} \big( p_i^2  + \alpha_i^2 \sigma^2 \textcolor{black}{\big)} \bigg]  \label{eq:ccuc_chebyshev_1a} \\
    & \hspace{10mm} p_i \leq \overline{P}_i u_i - \tilde \sigma_i \alpha_i, \quad \forall i \in \mathcal{I} \\
    & \hspace{10mm} - p_i \leq - \underline{P}_i u_i - \tilde \sigma_i \alpha_i, \quad \forall i \in \mathcal{I} \\    
    & \hspace{10mm} \textcolor{black}{\alpha_i \leq u_i, \quad \forall i \in \mathcal I}  \label{eq:ccuc_single_reformulation_alpha_2}  \\   
    & \hspace{10mm}\sum_{i \in \mathcal{I}} p_i  = D  - W \label{eq:ccuc_single_reformulation_power_balance2b}  \\   
    & \hspace{10mm} \sum_{i \in \mathcal{I}} \alpha_i= 1   \label{eq:ccuc_single_reformulation_integrality2b} \\     
    & \hspace{10mm} u_i = u_i^\ast, \quad \forall i \in \mathcal I \\
    & \hspace{10mm} p_i \geq 0, 0 \leq \alpha_i \leq 1, 0\leq u_i \leq 1, \forall i \in \mathcal I, 
\end{flalign}
\end{subequations}
\allowdisplaybreaks[0]
This hypothesis leads to the following theorem:

\begin{theorem}\label{theorem2} \normalfont Let  $\{p^\ast_i, \alpha^\ast_i, u^\ast_i, \forall i \in \mathcal{I}\}$ be an optimal solution of  \eqref{eq:ccuc_chebyshev} and let  $\{\lambda^\ast, \chi^\ast, \{\gamma^\ast_i, \forall i \in \mathcal{I}\}\}$  be dual variables of constraints (3d), (3e) and (13b) of the augmented equivalent in \eqref{eq:augumented_chebyshev}. Then $\{\{p^\ast_i, \alpha^\ast_i, u^\ast_i, \gamma^\ast_i,\forall i \in \mathcal{I}\}, \lambda^\ast, \chi^\ast, \}$ is a robust competitive equilibrium given by Definition~\ref{definition1}, i.e.:
\begin{enumerate}
    \item The market clears at $\sum_{i \in \mathcal{I}} \! p^\ast_i\! +\! W \!= \!D$ and $ \sum_{i \in \mathcal{I}}\! \alpha^\ast_i\!\!=\!\!1 $.
    \item Each producer maximizes its profit under the payment of $\Gamma_i^\ast =\lambda^\ast p_i^\ast + \chi^\ast  \alpha_i^\ast +  \gamma_i^\ast  u_i^\ast$.
\end{enumerate}
\end{theorem}

\begin{proof}  Consider \eqref{eq:ccuc_chebyshev}. If it is feasible and solved to optimality,  optimal values $p^\ast_i$ and $\alpha^\ast_i$ must satisfy equality constraints \eqref{eq:ccuc_single_reformulation_power_balance} and \eqref{eq:ccuc_single_reformulation_integrality}. As a result, it follows that $\sum_{i \in \mathcal{I}} \! p^\ast_i\! +\! W \!= \!D$ and $ \sum_{i \in \mathcal{I}}\! \alpha^\ast_i\!\!=\!\!1 $, i.e., the first postulate of Theorem~\ref{theorem2} holds.

Proving the second postulate of  Theorem~\ref{theorem2} requires showing that dual variables of the augmented optimization in \eqref{eq:augumented_chebyshev} represent and can be interpreted as marginal sensitivities of the equivalent constraints in the mixed-integer optimization in  \eqref{eq:ccuc_chebyshev}. This proof follows from  \cite[Proposition 1]{KUANG2019123}, which establishes equivalence between the optimal solution of a given MIQP problem and its augmented problem with relaxed integer decision set to their optimal values. Hence, dual variables $\lambda^\ast$, $\chi^\ast$, and $\gamma^\ast_i$ of the augmented problem in \eqref{eq:augumented_chebyshev} are sensitivities of the equivalent constraints in  \eqref{eq:ccuc_chebyshev}. 

Now we show that optimal values $p^\ast_i, \alpha^\ast_i, u^\ast_i, \lambda^\ast, \chi^\ast, \gamma^\ast_i$  maximize the profit of each producer. To this end, we recast \eqref{eq:augumented_chebyshev} as the following equivalent MISOC program using substitution $x_i = p_i^2$ and  $z_i = \alpha_i^2$: 
\allowdisplaybreaks[0]
\begin{subequations}
\label{eq:ccuc_chebyshev_socp}
\begin{flalign}
        &\hspace{-10mm} \min_{p_i, \alpha_i, u_i, x_i, z_i} \sum_{i \in \mathcal{I}} \bigg[C_{0,i} u_{i} +  C_{1,i} p_i  +   C^{}_{2,i} (x_i + \sigma^2z_i  )\bigg] \label{eq:ccuc_chebyshev_socp_obj} \\
  (\phi_i):   & \hspace{5mm} p_i^2 \leq x_i, \quad \forall i \in \mathcal{I} \label{eq:misocp_x} \\    
  (\psi_i):   & \hspace{5mm} \alpha_i^2 \leq z_i, \quad \forall i \in \mathcal{I}  \label{eq:misocp_z} \\
  (\overline{\mu}_i):   & \hspace{5mm} p_i \leq \overline{P}_i  u_i- \tilde \sigma_i \alpha_i, \quad \forall i \in \mathcal{I}  \\    
  (\underline{\mu}_i):   & \hspace{5mm}  -p_i \leq- \underline{P}_i  u_i - \tilde \sigma_i \alpha_i, \quad \forall i \in \mathcal{I}    \\
   \textcolor{black}{(\upsilon_i)}:    & \hspace{5mm} \textcolor{black}{\alpha_i \leq u_i, \quad \forall i \in \mathcal I}    \\   
  (\lambda):   & \hspace{5mm}   \sum_{i \in \mathcal{I}} p_i = D -W \label{eq:ccuc_chebyshev_socp_pwr_balance}    \\   
  (\chi):   & \hspace{5mm} \sum_{i \in \mathcal{I}} \alpha_i= 1    \label{eq:ccuc_chebyshev_socp_regulation}  \\ 
  (\gamma_i):   & \hspace{5mm} u_i=u_i^*, \quad \forall i \in \mathcal I  \label{eq:ccuc_chebyshev_socp_augmented}  \\     
	& \hspace{5mm} p_{i} \geq 0, \alpha_i \geq 0, 0 \leq u_i \leq 1, \quad \forall i \in \mathcal{I},   
\end{flalign}
\end{subequations}
\allowdisplaybreaks[0]
where   $x_i$ and $z_i$ are auxiliary decision variables, \eqref{eq:misocp_x}  and \eqref{eq:misocp_z} are SOC constraints. Note that dual variables of constraints in \eqref{eq:ccuc_chebyshev_socp} are given in parenthesis. 

In turn, the optimization in \eqref{eq:ccuc_chebyshev_socp} can be reformulated as the following equilibrium problem (as proven in Appendix~\ref{app:KKT_conditions}):
\allowdisplaybreaks
\begin{subequations} \label{eq:ep_gen}
\begin{flalign}
\bigg\{   \max_{p_i, u_i, \alpha_i, x_i, z_i } & \Pi_i  \label{eq:ep_gen1} \\
(\phi_i): \quad & p_i^2 \leq x_i  \label{eq:ep_gen2}\\    
(\psi_i): \quad  &\alpha_i^2 \leq z_i  \label{eq:ep_gen3}\\
(\overline{\mu}_i): \quad & p_i \leq \overline{P}_i  u_i- \tilde \sigma_i \alpha_i  \label{eq:ep_gen4}\\    
(\underline{\mu}_i): \quad  &  -p_i \leq -\underline{P}_i u_i - \tilde \sigma_i \alpha_i   \label{eq:ep_gen5} \\
(\upsilon_i):  \quad     &  \textcolor{black}{\alpha_i \leq u_i}  \bigg\}, \quad \forall i \in \mathcal{I}  \label{eq:ep_gen6}\\
\bigg\{ (\lambda) : \quad  &   \sum_{i \in \mathcal{I}} p_i  = D-W,   \label{eq:ep_gen7} \\   
(\chi) :  \quad &\sum_{i \in \mathcal{I}} \alpha_i = 1,    \label{eq:ep_gen8}   \\
(\mu_i) :  \quad &   u_i = u_i^*, \label{eq:ep_gen8b}  \bigg\}, 
\end{flalign}
\allowdisplaybreaks[0] 
\end{subequations}where \eqref{eq:ep_gen1}-\eqref{eq:ep_gen6} is solved by each producer individually and \eqref{eq:ep_gen7}-\eqref{eq:ep_gen8b} is solved by the market. Note that the objective function of each producer given by \eqref{eq:ep_gen1}  is profit-maximizing and is formulated based on Definition~\ref{definition1} as  $\Pi_i = \big( \lambda p_i + \chi  \alpha_i  + \gamma_i  u_i   - C_{0,i} u_{i} -  C_{1,i} p_i  -   C^{}_{2,i} (x_i + \sigma^2z_i)  \big)$. Each producer solves its optimization given by \eqref{eq:ep_gen1}-\eqref{eq:ep_gen5} and obtains optimal decisions $p_i^\prime, \alpha_i^\prime, u_i^\prime$ that must satisfy the market problem in \eqref{eq:ep_gen7}-\eqref{eq:ep_gen8b}, i.e.,  $\sum_{i \in \mathcal{I}} p_i^\prime + W = D$ and $\sum_{i \in \mathcal{I}} \alpha_i^\prime= 1$, which returns prices $\lambda^\prime$ and $\chi^\prime$. Under this equilibrium solution, the profit of each producer is maximized, due to the objective function in \eqref{eq:ep_gen1}, and  can be computed as $\Pi_i^\prime = \big( \lambda^\prime p_i^\prime + \chi^\prime  \alpha_i^\prime  + \gamma_i^\prime  u_i^\prime   - C_{0,i} u_{i}^\prime -  C_{1,i} p_i^\prime  -  C_{2,i} (p_i^\prime)^2 -C_{2,i} \sigma^2 (\alpha_i^\prime)^2  \big)$. 

Since the equilibrium problem in \eqref{eq:ep_gen} is equivalent to \eqref{eq:ccuc_chebyshev_socp}, as proven in Appendix~\ref{app:KKT_conditions}, and \eqref{eq:ccuc_chebyshev_socp} is equivalent to the original optimization in \eqref{eq:ccuc_chebyshev}, as per  \cite[Proposition 1]{KUANG2019123}, their optimal solutions are equal. Hence, we note   $\lambda^\ast = \lambda^\prime$, $\chi^\ast = \chi^{\prime}$, $\gamma_i^\ast = \gamma_i^\prime$, $p_i^\ast = p_i^\prime$, $\alpha_i^\ast = \alpha_i^\prime$, and $u_i^\ast = u_i^\prime$. This  leads to $\Pi_i^\ast =\Pi_i^\prime$. Since $\Pi_i^\prime$ is maximized by the optimization in \eqref{eq:ep_gen1}-\eqref{eq:ep_gen6}, so is $\Pi_i^\ast$. Thus,  $\{p^\ast_i, \alpha^\ast_i, u^\ast_i, \lambda^\ast, \chi^\ast, \gamma^\ast_i\}$  ensures that the second postulate of Theorem~\ref{theorem2} holds. 
\end{proof}

While semantically similar to Theorem~\ref{theorem1}, the result of  Theorem~\ref{theorem2} is a generalization of  Theorem~\ref{theorem1} that leverages SOC duality for electricity pricing and allows for more accurate market prices and dispatch allocations due to  (i) modeling  distributionally robust chance constraints in \eqref{eq:ccuc_chebyshev_socp}  (the assumption of $\bm{\omega}\sim N(0,\sigma^2)$ used in Theorem~\ref{theorem1} is no longer required) and (ii) considering quadratic production costs since $C_{2,i}>0$.  Accordingly, using Theorem~\ref{theorem2}, we can obtain  explicit expressions for energy, reserve and commitment prices by using stationary conditions of \eqref{eq:ccuc_chebyshev_socp} given in Appendix~\ref{app:KKT_conditions}. Indeed, re-arranging terms in   \eqref{eq:kkt_1}-\eqref{eq:kkt_3} leads to: 
\begin{subequations}
\begin{flalign}
& \lambda  = C_{1,i} + 2 C_{2,i} p_i + \overline{\mu}_i - \underline{\mu}_i \label{eq:chebkkt_1} \\
&\chi  = 2 C_{2,i} \sigma^2 \alpha_i + \overline{\mu}_i \tilde \sigma_i + \underline{\mu}_i \tilde \sigma_i - \textcolor{black}{\upsilon_i}  \label{eq:chebkkt_2}  \\
&\gamma_i = C_{0,i} - \overline{\mu}_i \overline{P}_i +\underline{\mu}_i \underline{P}_i + \textcolor{black}{\upsilon_i} \label{eq:chebkkt_3}.
\end{flalign}
Expressing $p_i$ and $\alpha_i$ from \eqref{eq:chebkkt_1}  and \eqref{eq:chebkkt_2} as functions of $\lambda$ and $\chi$ and plugging these expressions into \eqref{eq:ccuc_chebyshev_socp_pwr_balance} and \eqref{eq:ccuc_chebyshev_socp_regulation}, respectively, leads to the following expressions: 
\allowdisplaybreaks
\begin{flalign}
& \lambda = \bigg[ D-W +\sum_{i \in \mathcal{I}} \frac{C_{1,i} + \overline \mu_i - 
\underline \mu_{i}}{2C_{2,i}} \bigg]  / \sum_{i \in \mathcal {I}} \frac{1}{2C_{2,i}} \label{eq:chebkkt_a}\\ 
 & \chi = \bigg[  1 + \sum_{i \in \mathcal{I}} \frac{(\overline \mu_i + \underline \mu_i) \tilde \sigma_i - \textcolor{black}{\upsilon_i} }{2 C_{2,i} \sigma^2 } \bigg] / \sum_{i \in \mathcal {I}} \frac{1}{2C_{2,i} \sigma^2}. \label{eq:chebkkt_b}
\end{flalign}
\end{subequations}
\allowdisplaybreaks[0]
Similarly to \eqref{eq:ccuc_lp_price_expressions2}, $\lambda$ and $\gamma_i$ in \eqref{eq:chebkkt_3} and \eqref{eq:chebkkt_a} do not depend on uncertainty and risks parameters, while  $\chi$ in \eqref{eq:chebkkt_b} internalizes these parameters via $\tilde \sigma_i$.

\subsubsection{Pricing under the exact SOC reformulation}

Similarly to \eqref{eq:ccuc_chebyshev}, the optimization in  \eqref{eq:ccuc_drcc} is a MISOC problem and, therefore, we can follow the same procedure as described in Section~\ref{sec:pricing:chebyshev} to show that  \eqref{eq:ccuc_drcc} can yield a robust competitive equilibrium as given by Definition~\ref{definition1}.  First, we define the continuous equivalent of \eqref{eq:ccuc_drcc}:
\allowdisplaybreaks
\begin{subequations}
\label{eq:ccuc_drcc_eq}
\begin{flalign}
  &\!\!\!\!\!\!\!\!\!  \!\!\!\!\!\!\!\!\! \min_{p_i, \alpha_i, u_i, y_i, x_i, z_i, \pi_i}\!  \sum_{i \in \mathcal{I}}\! \bigg[C_{0,i} u_{i} \!+ \! C_{1,i} p_i\! +\!   C_{2,i} (x_i \!+\! z_i \sigma^2 ) \bigg] \\
(\phi_i)  : & \hspace{5mm} p_i^2 \leq x_i, \quad \forall i \in \mathcal{I} \label{eq:ccuc_drcc_eq_1}\\    
(\psi_i) : & \hspace{5mm}  \alpha_i^2 \leq z_i, \quad \forall i \in \mathcal{I} \label{eq:ccuc_drcc_eq_2} \\    
 (\overline{\rho}_i) :    & \hspace{5mm} p_i - \frac{\overline{P}_i + \underline{P}_i}{2} u_i \leq y_i + \pi_i, \quad \forall i \in \mathcal{I}  \\
  (\underline{\rho}_i) :     & \hspace{5mm} -p_i + \frac{\overline{P}_i + \underline{P}_i}{2} u_i \leq y_i + \pi_i, \quad \forall i \in \mathcal{I}  \\
    (\rho_i):  & \hspace{5mm} y_i^2 + \alpha_i^2 \sigma^2\leq \epsilon_i \bigg( \frac{\overline{P}_i - \underline{P}_i}{2} - \pi_i \bigg)^2, \quad \forall i \in \mathcal{I}  \label{eq:ccuc_drcc_eq_3} \\
      \textcolor{black}{(\upsilon_i):}   &   \hspace{5mm}  \textcolor{black}{\alpha_i \leq u_i} \label{eq:ccuc_drcc_eq_3b}  \\   
    (\lambda):  & \hspace{5mm} \sum_{i \in \mathcal{I}} p_i  = D-W  \label{eq:ccuc_drcc_eq_powerbalance} \\   
      (\chi):  & \hspace{5mm} \sum_{i \in \mathcal{I}} \alpha_i= 1  \label{eq:ccuc_drcc_eq_regulation} \\   
            (\gamma_i):  & \hspace{5mm} u_i= u_i^\ast  \label{eq:ccuc_drcc_eq_status} \\     
	& \hspace{-7mm} p_{i} \!\geq\! 0, \alpha_i\! \geq\! 0, y_{i} \!\geq\! 0, 0\! \leq\! \pi_{i}\! \leq\! \frac{\overline{P}_i\! -\! \underline{P}_i}{2}, 0 \!\leq\! u_i\! \leq \! 1,  \forall i \in \mathcal{I},   
\end{flalign}
\end{subequations}
\allowdisplaybreaks[0]
where $x_i = p_i^2$ and  $z_i = \alpha_i^2$ are auxiliary variables, \eqref{eq:ccuc_drcc_eq_1} and \eqref{eq:ccuc_drcc_eq_1} are auxiliary SOC constraints, and $u_i^*$ is the optimal solution of \eqref{eq:ccuc_drcc} that can be obtained using off-the-shelf solvers. Using the original mixed-integer optimization in \eqref{eq:ccuc_drcc} and  its augmented SOC equivalent in \eqref{eq:ccuc_drcc_eq}, we prove:

\begin{theorem}\label{theorem3} \normalfont Let  $\{p^\ast_i, \alpha^\ast_i, u^\ast_i, \forall i \in \mathcal{I}\}$ be an optimal solution of \eqref{eq:ccuc_drcc} and let  $\{\lambda^\ast, \chi^\ast, \{\gamma^\ast_i, \forall i \in \mathcal{I}\}\}$ be dual variables of constraints \eqref{eq:ccuc_drcc_eq_powerbalance}, \eqref{eq:ccuc_drcc_eq_regulation}  and \eqref{eq:ccuc_drcc_eq_status} of the augmented SOC equivalent in \eqref{eq:ccuc_drcc_eq}. Then $\{\{p^\ast_i, \alpha^\ast_i, u^\ast_i, \gamma^\ast_i,\forall i \in \mathcal{I}\}, \lambda^\ast, \chi^\ast\}$ is a robust competitive equilibrium given by Definition~\eqref{definition1}, i.e.:
\begin{enumerate}
    \item The market clears at $\sum_{i \in \mathcal{I}} \!p^\ast_i\! =\! D \!-\!W$ and $ \sum_{i \in \mathcal{I}} \!\alpha^\ast_i\!=\! 1 $.
    \item Each producer maximizes its profit under the payment of $\Gamma_i^\ast = \lambda^\ast p_i^\ast + \chi^\ast  \alpha_i^\ast +  \gamma_i^\ast  u_i^\ast$.
\end{enumerate}
\end{theorem}
\begin{proof}
Note that \eqref{eq:ccuc_drcc} and  \eqref{eq:ccuc_drcc_eq} are MISOC and SOC problems, and, thus, are similar to  \eqref{eq:ccuc_chebyshev}  and  \eqref{eq:ccuc_chebyshev_socp} in Theorem~\ref{theorem2}. Therefore, Theorem~\ref{theorem3} can be proven analogously to the proof of Theorem~\ref{theorem2}. We omit the proof for brevity.
\end{proof}

Using the equilibrium established by  Theorem~\ref{theorem3}, we can analyze the dependency of the resulting prices on uncertainty and risk parameters. Consider the Lagrangian function of \eqref{eq:ccuc_drcc_eq} and  recall that $x_i = p_i^2$ and  $z_i = \alpha_i^2$:
\begin{flalign}
\mathcal{L}& = \sum_{i \in \mathcal{I}}\! \bigg[C_{0,i} u_{i} \!+ \! C_{1,i} p_i\! +\!   C_{2,i} (p_i^2 \!+\! \alpha_i^2 \sigma^2 ) \!+\! \overline{\rho}_i \!  \big( \!p_i\! -\! \frac{\overline{P}_i + \underline{P}_{\textcolor{black}{i}}}{2} u_i    \nonumber &  \\ &  - y_i - \pi_i  \big)     \! +\! \underline{\rho}_i   \big( \frac{\overline{P}_i + \underline{P}_i}{2} u_i -p_i  - y_i - \pi_i     \big)   \!+\! \rho_i   \bigg( y_i^2 + \alpha_i^2 \sigma^2 \nonumber &   \\  & -\epsilon_i \big( \frac{\overline{P}_i - \underline{P}_i}{2} - \pi_i \big)^2  \bigg)  \!+ \! \textcolor{black}{\upsilon_i ( u_i - \alpha_i)} \!+ \gamma_i\! (u_i^\ast - u_i) \bigg]   \nonumber &  \\ &  +\lambda \big( \!D -\!   W \!-\! \sum_{i \in \mathcal{I}} p_i\big)    + \chi \big(1\!-\!\sum_{i \in \mathcal{I}} \alpha_i \big)  &                    
\end{flalign}
and obtain the following stationary conditions: 
\begin{subequations}
\begin{flalign}
&\frac{\partial \mathcal{L}}{\partial p_i}= C_{1,i} + 2 C_{2,i} p_i + \overline{\rho}_i   - \underline{\rho}_i   - \lambda = 0 \label{eq:dual_prices_drcc_eq1a}  \\
&\frac{\partial \mathcal{L}}{\partial \alpha_i}= 2 C_{2,i} \sigma^2 \alpha_i +   2\rho_i \sigma^2 \alpha_i  - \chi - \textcolor{black}{\upsilon_i} = 0.  \label{eq:dual_prices_drcc_eq1b}  \\
&\frac{\partial \mathcal{L}}{\partial u_i}= C_{0,i} + \frac{\overline{P}_i + \underline{P}_i}{2} (\underline{\rho}_i - \overline{\rho}_i) +\textcolor{black}{\upsilon_i} -\gamma_i.  \label{eq:dual_prices_drcc_eq1c} 
\end{flalign}
\end{subequations}
Expressing $p_i$ and $\alpha_i$ from \eqref{eq:dual_prices_drcc_eq1a} and \eqref{eq:dual_prices_drcc_eq1b} as functions of $\lambda$ and $\chi$, respectively,  and using these expressions in  \eqref{eq:ccuc_drcc_eq_powerbalance} and \eqref{eq:ccuc_drcc_eq_regulation} leads to:
\begin{subequations} \label{eq:exact_prices}
\begin{flalign}
& \lambda = \bigg[ D - W - \sum_{i \in \mathcal I } \frac{(\underline{\rho}_i - \overline{\rho}_i - C_{1,i})}{2 C_{2,i}} \bigg]/ \sum_{i \in \mathcal{I}} \frac{1}{2 C_{2,i}} \label{eq:exact_prices_lambda}  \\
& \chi = \frac{1}{\sum_{i \in \mathcal I} 1/ (2 C_{2,i} \sigma^2  + 2\rho_i \sigma^2) } - \textcolor{black}{v_i}.  \label{eq:exact_prices_chi}
\end{flalign}
Using  \eqref{eq:dual_prices_drcc_eq1c}, we obtain:
\begin{flalign}
\gamma_i = C_{0,i} + \frac{\overline{P}_i + \underline{P}_i}{2} (\underline{\rho}_i - \overline{\rho}_i) + \textcolor{black}{\upsilon_i}.\label{eq:exact_prices_gamma}
\end{flalign}

\end{subequations}

Note that similarly to the prices in \eqref{eq:exact_prices_lambda} and \eqref{eq:exact_prices_gamma},  $\lambda$ and $\gamma_i$ in  \eqref{eq:exact_prices_lambda} and \eqref{eq:exact_prices_gamma} are independent of uncertainty and risk parameters, while $\chi$  in \eqref{eq:exact_prices_chi} internalizes $\sigma$.

\subsection{Design Properties of the SOC-based Markets}

The market outcomes obtained under Theorems~\ref{theorem2} and \ref{theorem3} not only internalize uncertainty and risk parameters in the price formation process, but also are helpful in ensuring such market design properties as cost recovery and revenue adequacy.

\subsubsection{Cost recovery} Cost recovery implies that  producers recover their operating cost from market outcomes and can be formalized as $\Pi_i \geq 0, \forall i \in \mathcal I$\textcolor{black}{, where $\Pi_i = \big( \lambda p_i + \chi  \alpha_i  + \gamma_i  u_i   - C_{0,i} u_{i} -  C_{1,i} p_i  -   C^{}_{2,i} (x_i + \sigma^2z_i)  \big)$ as defined for the equilibrium problem in \eqref{eq:ep_gen}. Since the optimization problem of each producer in \eqref{eq:ep_gen1}-\eqref{eq:ep_gen6} is convex, we can invoke  the strong duality theorem for the optimal market outcomes. The strong duality theorem makes it possible to equate the primal and dual objective functions of \eqref{eq:ep_gen1}-\eqref{eq:ep_gen6} as follows:}
\allowdisplaybreaks
\begin{subequations} \label{eq:cost_recovery}
\begin{flalign}
\Pi_i = \gamma_i u_i^\ast,  \label{eq:cost_recovery1}
\end{flalign}
\allowdisplaybreaks[0]
\textcolor{black}{where $\Pi_i$ is the primal objective function and $\gamma_i u_i^\ast$ is the dual objective function.} At the optimum,  \eqref{eq:ccuc_chebyshev_socp} and \eqref{eq:ep_gen} yield equivalent solutions (see Appendix~\ref{app:KKT_conditions}). Therefore, we express $\gamma_i^\ast$ from \eqref{eq:chebkkt_3} and recast \eqref{eq:cost_recovery1} as follows:
\allowdisplaybreaks
\begin{flalign}
\Pi_i^\ast = C_{0,i} u_i^\ast + \upsilon^\ast_i u_i^\ast  + \underline \mu_i^\ast \underline P_i u_i^\ast - \overline \mu_i^\ast \overline P_i u_i^\ast.  \label{eq:cost_recovery2}
\end{flalign}
\allowdisplaybreaks[0]
Since \eqref{eq:ep_gen1}-\eqref{eq:ep_gen6} is a maximization problem,  $\overline \mu_i^\ast \geq 0$, $\underline \mu_i^\ast \geq 0$ and \textcolor{black}{$\upsilon^\ast_i \geq 0$}. Therefore, the first three terms of \textcolor{black}{\eqref{eq:cost_recovery2}} are always non-negative. Since $\underline \mu_i^\ast \geq 0$,  the last term in  \textcolor{black}{\eqref{eq:cost_recovery2}} is negative and, thus, $\Pi^\ast_i$ can attain negative values. However, there are two specific, but practical cases to ensure that $\Pi^\ast_i \geq 0$. First, we can restrict $\gamma^\ast_i \geq 0$ to avoid negative (confiscatory) commitment prices, as common in real-life markets (see \cite{SIOSHANSI2014143}). This will make all terms in \textcolor{black}{\eqref{eq:cost_recovery2}} non-negative and guarantee that $\Pi^\ast_i \geq 0$.  Second, if the market is convex, i.e., controllable generators have $\underline P_i =0$ and $p_i \in [0, \overline P_i]$, the cost recovery is guaranteed: 
\begin{flalign}
\Pi_i^\ast \geq  0  \label{eq:cost_recovery21}
\end{flalign}
 This is an important property of the proposed SOC-based market design that allows for explicitly considering uncertainty and risk parameters in the price formation process without compromising social welfare.

Analogously, in the case of the exact SOC reformulation,  \eqref{eq:ccuc_drcc} can be used to formulate an equilibrium problem similar to \eqref{eq:ep_gen}. In this equilibrium problem, each producer is modeled as $\{ \max \Pi_i | \{\text{Eq.~} \eqref{eq:ccuc_drcc_eq_1}-\eqref{eq:ccuc_drcc_eq_3}  \} \}$. Hence, \textcolor{black}{ similarly to the Chebyshev case,} we can exploit the strong duality property to obtain: 
\begin{flalign}
\Pi_i^\ast = \rho_i^\ast \epsilon_i \big(\frac{\overline{P}_i - \underline{P}_i}{2}\big)^2+\gamma_i^\ast u_i^*. 
 \label{eq:cost_recovery3}
\end{flalign}
\end{subequations}
Since $\rho^\ast_i \geq 0$, $(\overline{P}_i - \underline{P}_i) \geq 0$, we can ensure that $\Pi_i^\ast \geq 0$ in \eqref{eq:cost_recovery3} similarly to  \eqref{eq:cost_recovery21}, i.e., either  we restrict $\gamma_i^\ast \geq 0$ (see \cite{SIOSHANSI2014143}) or the market is convex and $p_i \in \big[ 0, \overline{P}_i\big]$. 

Note that if no additional restriction is imposed on non-negativity of dual variables $\underline{\mu}^\ast_i$ and $\gamma^\ast_i$, one can compute the uplift payment for each producer as $\Upsilon_i^\ast = \max\big[0, -\Pi_i^\ast \big]$, if $\Pi_i^\ast <0$.

\subsubsection{Revenue adequacy} Revenue adequacy is needed to ensure that the total payment from consumers collected by the market operator covers  the total payment to producers made by the market operator. Since the stochastic market designs in  Theorems~\ref{theorem2}-\ref{theorem3} are based on the same principles as the currently practiced market design in \cite{ONEILL2005269}, they are also revenue-inadequate. Thus,  the market revenue deficit ($\Delta^\ast$) is:
\allowdisplaybreaks
\begin{subequations}
\begin{flalign}
\Delta^\ast = -\min \big[ 0,  \sum_{i \in \mathcal I} \Gamma_i^\ast+ \lambda^\ast W -\lambda^\ast D      \big]   \label{eq:deficit},
\end{flalign}
\end{subequations}
\allowdisplaybreaks[0]
where the first two terms  represent the payment to controllable and wind power producers and the last term is the payment collected from consumers. Recall that  Theorems~\ref{theorem2}-\ref{theorem3} define  $\Gamma_i^\ast = \lambda^\ast p_i^\ast + \chi^\ast \alpha_i^\ast + \gamma_i^\ast u_i^\ast $ and establish that $\sum_{i \in \mathcal I} \alpha_i^\ast =1$ and  $\sum_{i \in \mathcal I} p^\ast_i =  (D-W)$. Hence, \eqref{eq:deficit} is recast as:
\allowdisplaybreaks
\begin{subequations}
\begin{flalign}
\Delta^\ast =  - \min \big[ 0, - \chi^\ast - \sum_{i \in \mathcal I} \gamma_i^\ast u_i^\ast \big]. \label{eq:deficit2}
\end{flalign}
\end{subequations}
\allowdisplaybreaks[0] 
Since $\chi^\ast \geq 0$, the sign of \eqref{eq:deficit2} depends on $\gamma_i^\ast$, which can attain both negative and positive values. Hence, if $\Delta^\ast \geq 0$ in \eqref{eq:deficit2}, this deficit must be additionally allocated among consumers, e.g., as in \cite{ONEILL2005269}. However, similarly to the cost recovery properties discussed above, we can guarantee $\Delta^\ast =0$, i.e. the market is revenue-adequate, in the special cases of non-confiscatory prices ($\gamma_i^\ast \leq 0$) and convex markets ($p_i \in \big[ 0, \overline{P}_i\big]$). Indeed, if $\gamma_i^\ast \geq 0$, then $(- \chi^\ast - \sum_{i \in \mathcal I} \gamma_i^\ast u_i^\ast) \leq 0$, which leads to $\Delta^\ast =0$. Similarly, if the market is convex $p_i \in [ 0, \overline{P}_i]$, we obtain from \eqref{eq:chebkkt_3} and \eqref{eq:exact_prices_gamma} that $\gamma^\ast_i \geq 0$, which results in $(- \chi^\ast - \sum_{i \in \mathcal I} \gamma_i^\ast u_i^\ast) \leq 0$ and $\Delta^\ast =0$.

\subsubsection{Expected vs Per Scenario Performance}

The cost recovery and revenue adequacy properties described above are shown for expected quantities, i.e., the assumption is that $\sum_{i \in \mathcal{I}} p_i^\ast = D-W$ and $\sum_{i \in \mathcal{I}} \alpha_i^\ast= 1$. However, since these two constraints are always met, if the optimizations in \eqref{eq:ccuc_chebyshev} and \eqref{eq:ccuc_drcc}  are feasible, we can invoke \cite[Lemma 2.1]{doi:10.1137/130910312}, which ensures that the expected solution is ``viable'', equivalent to the solution   for every realization of uncertainty assumed on random variable $\bm{\omega}$, e.g., $\sum_{i \in \mathcal I} p_i (\bm{\omega}) +W+\bm{\omega}=D$. \textcolor{black}{As a result, the cost recovery  and revenue adequacy properties of the market outcomes obtained with Theorems~\ref{theorem2}-\ref{theorem3} hold for both the expected case and every realization of $\bm{\omega}$ drawn consistently with the assumed uncertainty (e.g., distribution parameters or uncertainty set) and risk tolerance (e.g., tolerance to violating chance constraints).} Hence, unlike scenario-based stochastic programming \cite{8246578}, the proposed market designs do not require trading-off market outcomes among the expected and per scenario cases at the expense of increasing the operating cost.

\section{Case Study}

 The case study is carried out on the 8-zone ISO New England testbed \cite{7039273,Li2017An8I}, which includes 76 thermal generators with a total installed capacity of roughly 30 GW and techno-economic characteristics reported in \cite{7039273}. \textcolor{black}{The notable feature of wind power modeling in \cite{7039273} is that it adopts an agent-based approach to  account for  the  effects of
local weather conditions, changes in the mix of wind turbine
types, and changes in the geographical placement of wind
turbines in order to model future wind penetration outputs at the system level.} It is assumed that all nuclear power plants are committed ($\approx$ 8 GW) to serve base loads.  The forecast wind power output ($W$) is modeled as described in \cite{Li2017An8I} for three different penetration levels:  2\% (current), 10\% and 20\% of the total demand. \textcolor{black}{We additionally assume that $W\geq0$, i.e. the 
 system-wide wind power production is always non-negative},  as well as that $\sigma=0.2 W$ and  $ \epsilon= \epsilon_i, \forall i \in \mathcal I$, i.e., the market operator  has a  uniform tolerance to  constraint violations. \textcolor{black}{Furthermore, we vary the value of parameter $\epsilon$ in the range from 0.0001 to 0.25 in order to capture the sensitivity of market outcomes to a wide range of choices available for this parameter. In practice, each market operator will need to calibrate this value to match their security preferences given the specifics of the underlying transmission system.} \textcolor{black}{All numerical results presented below were computed on a 2.9 GHz Intel Core i5 with 8 GB RAM under macOS Mojave. Table~\ref{tab:comp_times} reports the size of each optimization problem solved. Note that each optimization problem, including the exact MISOC equivalent,  was solved under 3 seconds. }

\begin{table}[t]
\renewcommand{\arraystretch}{0.8}
\setlength\tabcolsep{3.5 pt}
\captionsetup{font={color=blue}}
\caption{Numbers of Variables and Constraints}
\label{tab:comp_times}
\centering
{\color{black}
\begin{tabular}{c|ccc}
\toprule
\# of&  MILP in \eqref{eq:ccuc} & Chebyshev in \eqref{eq:ccuc_chebyshev} &  Exact SOC in \eqref{eq:ccuc_drcc} \\
 \midrule
Binary variables      & $N_{\mathcal I}$ &$N_{\mathcal I}$ & $N_{\mathcal I}$ \\
                                                 \midrule
Continuous variables      & $2 N_{\mathcal I}$  & $2 N_{\mathcal I}$ & $4 N_{\mathcal I}$  \\
                                                 \midrule
Constraints      & $2 N_{\mathcal I} +2$ &  $2 N_{\mathcal I} +2$  &   $3 N_{\mathcal I} +2$  \\
\bottomrule
\end{tabular}
} \\
\vspace{2mm}
\textcolor{black}{$N_{\mathcal I}$ denotes the cardinality of set $\mathcal I$, i.e. $N_{\mathcal I} = card (N_{\mathcal I})$} 
\end{table}

\begin{figure}[t!]
  \centering

    \includegraphics[width=\linewidth]{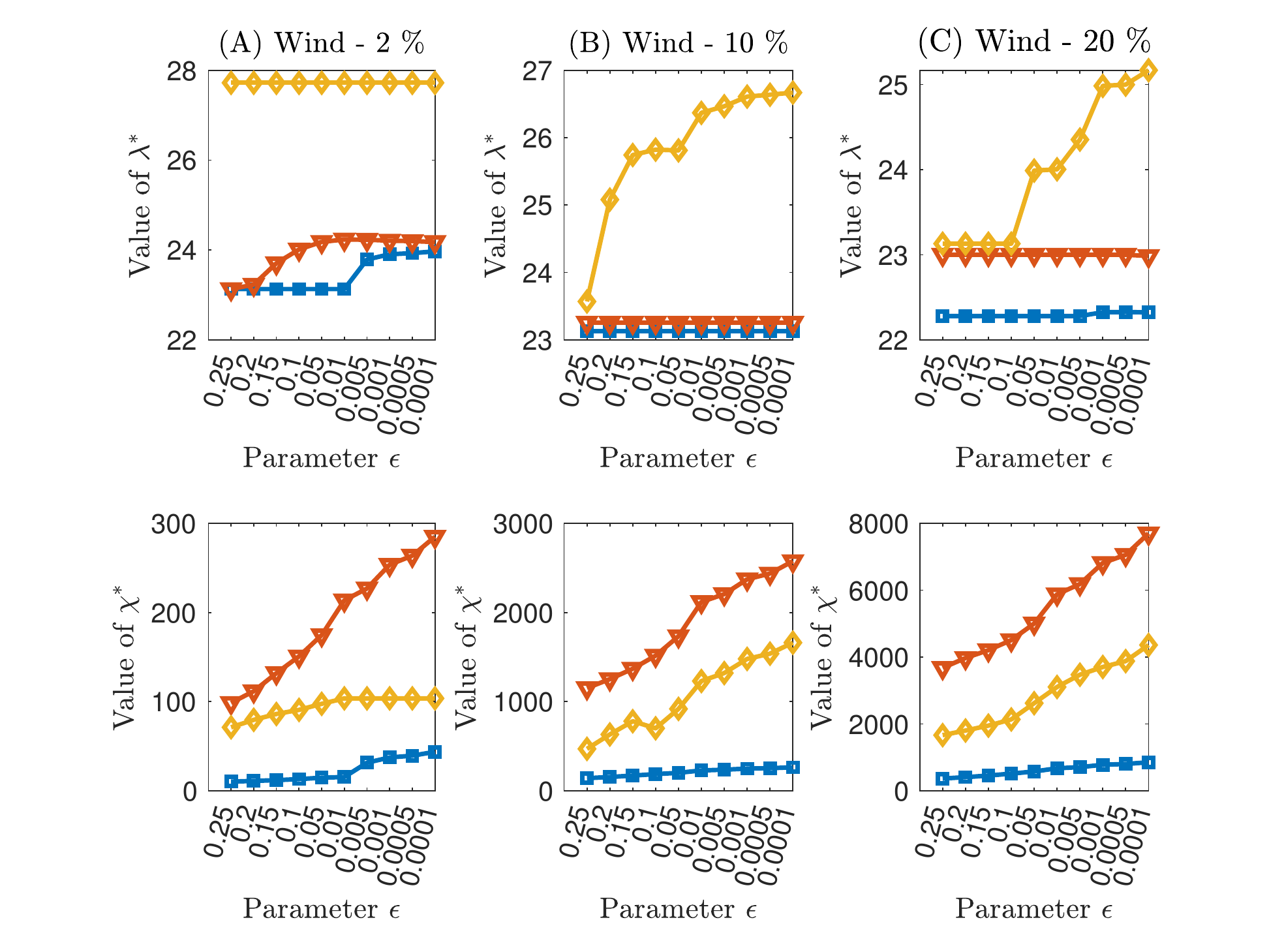}
    \includegraphics[width=0.4\linewidth]{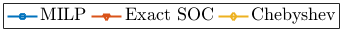}
       \caption{Energy ($\lambda^\ast$) and reserve ($\lambda^\ast$) prices for different values of $\epsilon$ under different market designs.}
    \label{fig:mprices}
    
\end{figure}

Figure~\ref{fig:mprices} compares the energy ($\lambda^\ast$) and reserve ($\chi^\ast$) prices obtained with the \textcolor{black}{MILP} model (Theorem~\ref{theorem1}), Chebyshev model (Theorem~\ref{theorem2}), exact MISOC reformulation (Theorem~\ref{theorem3})  for  different values of $\epsilon$. The effect of parameter $\epsilon$ on the resulting energy and reserve prices varies. Since the energy prices in all three models do not explicitly depend on $\epsilon$,  in some cases they remain constant for a wide range of values of $\epsilon$. Also, as the wind penetration rate increases, the energy prices tend to decrease for the same value of $\epsilon$, since more wind power generation replaces controllable generators with a relatively high production cost and the remaining controllable  generators are dispatched in an out-of-merit order.  On the other hand, as the value of $\epsilon$ reduces, reserve prices monotonically increase under all models and wind penetration rates, thus  reflecting a greater need in reserve to deal with the uncertainty and variability of wind power generation. In all simulations, the Chebyshev model, which is based on a conservative approximation of chance constraints, yields the greatest energy prices, regardless of the value of $\epsilon$ chosen.  On the other hand, the reserve prices under the Chebyshev approximation is lower than under the exact MISOC reformulation, since  the Chebyshev's conservative dispatch leads to  a greater out-of-merit order degree that results in a large amount of committed headroom capacity available  for providing reserves. 
\begin{figure}[t!]
  \centering

    \includegraphics[width=\linewidth]{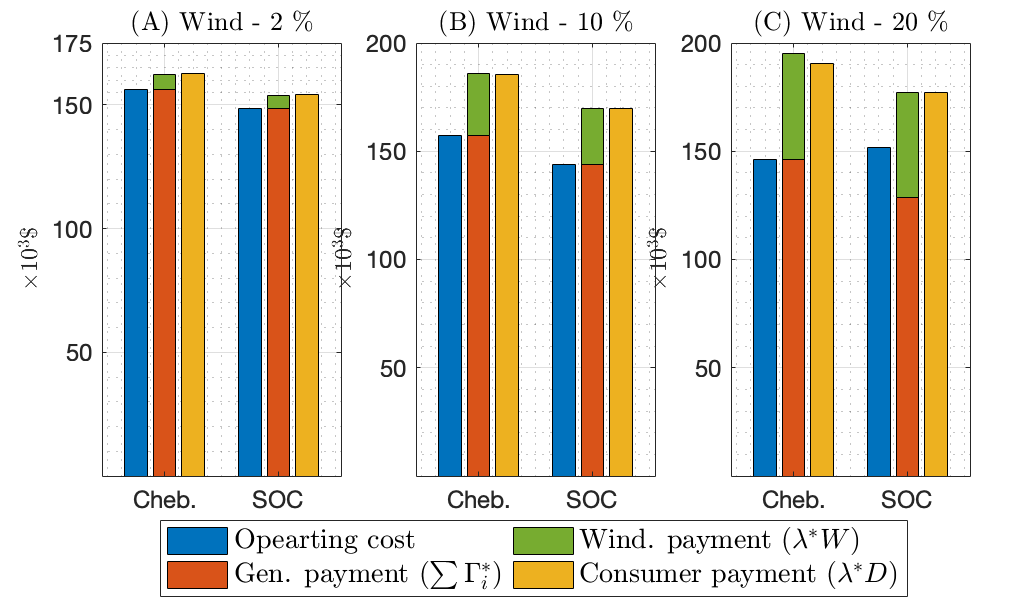}

       \caption{Comparison of the Chebyshev and MISOC models in terms of their market design features.}    
    \label{fig:revenue}
\end{figure}

\begin{figure}[b!]
  \centering

    \includegraphics[width=\linewidth]{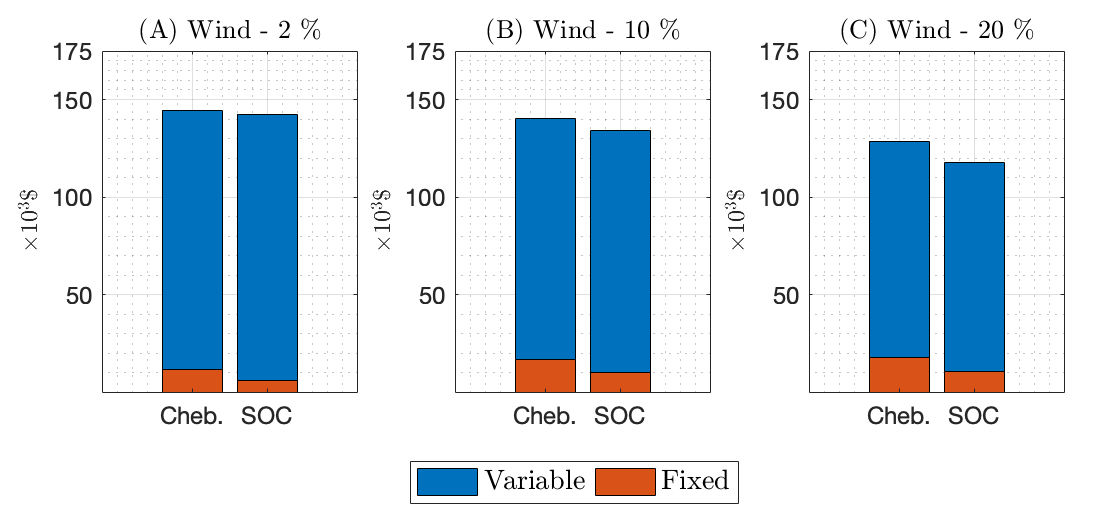}

       \caption{\textcolor{black}{Comparison of the Chebyshev and MISOC models in terms of their fixed and variable operating costs. In both models, the fixed operating cost is computed as $\sum_{i \in \mathcal I} C_{0,i} u_i$ and the variable operating cost is computed as the total operating cost minus the fixed operating cost.}}    
    \label{fig:cost}
\end{figure}

The trade-off between the energy and reserve prices obtained with the Chebyshev approximation and exact MISOC reformulation affects the revenue adequacy and cost recovery of these models.  Figure \ref{fig:revenue} compares the Chebyshev and MISOC models for $\epsilon=0.0001$, as it  is the most conservative solution among the results in Figure~\ref{fig:mprices} and, therefore, is expected to cause greatest out-of-merit order distortions. \textcolor{black}{Additionally, we} analyze the market performance of these models in terms of their total operating cost as defined by respective objective functions \textcolor{black}{and shown in  Figure \ref{fig:cost}}, total payment made by the market to controllable ($\sum_{i \in \mathcal I} \Gamma_i^\ast$ \textcolor{black}{as itemized in Figure \ref{fig:payment}}) and wind power ($\lambda^\ast W$) generators  and the total payment collected by the market from consumers ($\lambda^\ast D$).  Regardless of the wind penetration rate, the Chebyshev model yields a more expensive solution due to its inherent conservatism\textcolor{black}{, see Figure \ref{fig:cost}. This conservatism results in 311 MW of more committed power of conventional generators in the Chebyshev case relative to the MISOC case.}  Nevertheless, both the Chebyshev and exact MISOC models yield such prices that the total payment $\sum_{i \in \mathcal I} \Gamma^\ast_i$ is sufficient to cover the total operating cost, as well as we also   manually checked that $\Gamma_i^\ast \geq 0, \forall i \in \mathcal I$, i.e., every controllable generator attains a non-negative profit. In other words, the market outcomes in Figure \ref{fig:mprices} lead to cost recovery for all producers. 

As the wind penetration rate increases, we observe in Figure~\ref{fig:revenue} that the payment made by the market to wind power generators increases. Although the payment made to controllable generators reduces, the effect of zero-cost wind power generators suppress electricity prices at higher wind penetration rates (see Figure~\ref{fig:mprices}), which makes the market revenue-inadequate (e.g., the market deficit $\Delta^\ast >0$). \textcolor{black}{This revenue inadequacy is observed for the 20\% penetration levels and causes the relative mismatch between the total payment to producers  and the total payment from consumers equal to 0.2\% and 3.1\% for the MISOC and Chebyshev cases, respectively.}

\textcolor{black}{Furthermore, the effect of  greater wind penetration rates is observed in Figures \ref{fig:cost} and \ref{fig:payment}. Thus, greater wind penetrations tend to increase fixed costs in absolute values and relative  to the total operating cost for both the MISOC and Chebyshev market designs. However, the fixed costs of  the MISOC solution is systematically lower than in the Chebyshev case. Figure \ref{fig:payment} shows that the reserve and commitment payments increase for greater wind penetration rates. Notably, the MISOC market design consistently results in lower commitment payments than the Chebyshev case. }


\begin{figure}[t!]
  \centering

    \includegraphics[width=\linewidth]{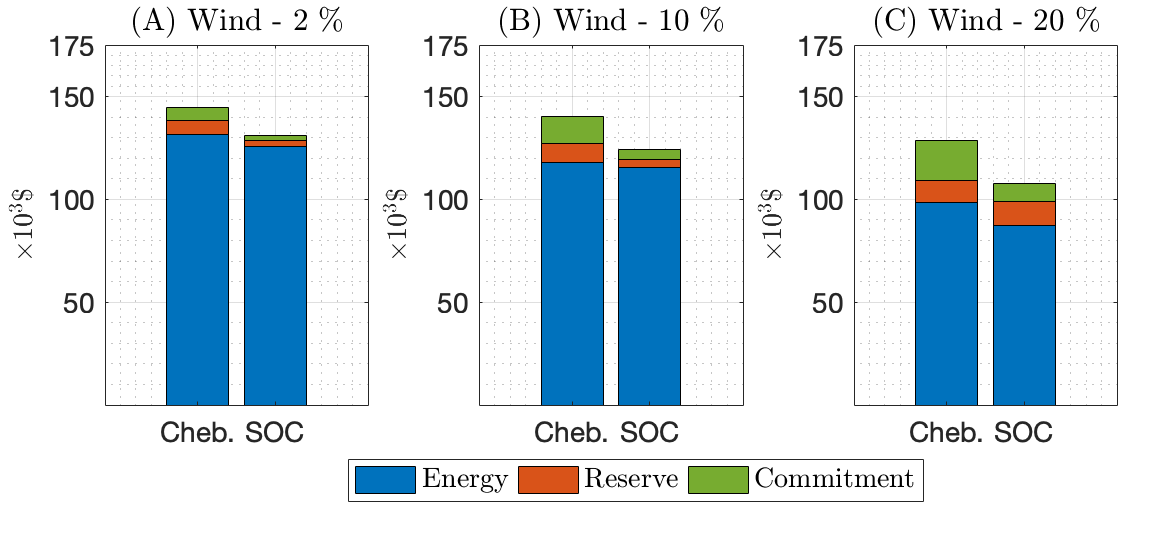}

       \caption{\textcolor{black}{Comparison of the Chebyshev and MISOC models in terms of the payment to conventional generators. The energy, reserve and commitment payments are defined as $ \sum_{i \in \mathcal I}  \lambda p_i $, $ \sum_{i \in \mathcal I} \chi \alpha_i$ and  $\sum_{i \in \mathcal I}  \gamma_i u_i$.}}    
    \label{fig:payment}
\end{figure}

\section{Conclusion}
This paper described an alternative approach to design a stochastic wholesale electricity market that allows one to internalize uncertainty of renewable generation resources and risk tolerance of the market operator in the price formation process using the chance constraints. The resulting stochastic market design exploits SOC duality to obtain a robust competitive  equilibrium that has the cost recovery and revenue adequacy properties similar to existing deterministic markets. \textcolor{black}{In the future, our work will focus on the application of the proposed pricing theory to multi-period network- and security-constrained stochastic market designs, which are needed by current market practices, and on achieving revenue adequacy of the stochastic market (e.g. by means of using alternative auction schemes \cite{8651383})}.

\begin{appendices} \section{Equivalence of \eqref{eq:ccuc_chebyshev_socp}  and \eqref{eq:ep_gen}}
\label{app:KKT_conditions}

We prove that \eqref{eq:ccuc_chebyshev_socp}  and \eqref{eq:ep_gen} yield equivalent solutions.  Consider \eqref{eq:ccuc_chebyshev_socp}  and recall that $x_i = p_i^2$ and $z_i = \alpha_i^2$. The  Lagrange function of \eqref{eq:ccuc_chebyshev_socp} is  then given by:
\allowdisplaybreaks
\begin{subequations}
\begin{flalign}
 \nonumber  \mathcal  L  = & \sum_{i \in \mathcal I}\bigg[C_{0,i} u_{i} +  C_{1,i} p_i  +   C_{2,i} (p^2_i + \sigma^2\alpha^2_i  ) + 
 \gamma_i (u_i^\ast - u_i) \nonumber \\ &+ \overline{\mu}_i  (p_i - \overline{P}_i  u_i+ \tilde \sigma_i \alpha_i)     + \underline{\mu}_i (- p_i + \underline{P}_i  u_i +   \tilde \sigma_i \alpha_i) \bigg] \nonumber \\ & + \lambda \big( D-W - \sum_{i \in \mathcal{I}} p_i\big) + \chi \big( 1-\sum_{i \in \mathcal{I}} \alpha_i \big) \label{eq:lagrangian}
\end{flalign}
\end{subequations}
\allowdisplaybreaks[0]
Using \eqref{eq:lagrangian}, we obtain the KKT conditions of  \eqref{eq:ccuc_chebyshev_socp} :
\allowdisplaybreaks
\begin{subequations}
\begin{flalign}
& \frac{\partial \mathcal L}{\partial p_i} = C_{1,i} + 2 C_{2,i} p_i + \overline{\mu}_i - \underline{\mu}_i  - \lambda = 0 \label{eq:kkt_1} \\
&\frac{\partial \mathcal L}{\partial \alpha_i} = 2 C_{2,i}  \sigma^2 \alpha_i + \overline{\mu}_i \tilde \sigma_i + \underline{\mu}_i \tilde \sigma_i  - \chi = 0 \label{eq:kkt_2}  \\
&\frac{\partial \mathcal L}{\partial u_i} = C_{0,i} - \gamma_i -\overline{\mu}_i \overline{P}_i  +\underline{\mu}_i \underline{P}_i = 0  \label{eq:kkt_3}  \\
& \sum_{i \in \mathcal I} p_i  = D-W \label{eq:kkt_4}  \\
& \sum_{i \in \mathcal I} \alpha_i = 1 \label{eq:kkt_5} \\
& u_i = u_i^\ast  \label{eq:kkt_6}  \\
& 0 \leq \overline{\mu}_i \perp ( \overline{P}_i  u_i- \tilde \sigma \alpha_i - p_i) \geq 0 \label{eq:kkt_7}  \\    
& 0 \leq \underline{\mu}_i \perp ( p_i - \underline{P}_i  u_i- \tilde \sigma \alpha_i) \geq 0, \label{eq:kkt_8} 
\end{flalign}
\end{subequations}
\allowdisplaybreaks[0]
where \eqref{eq:kkt_1}-\eqref{eq:kkt_2} are the stationary conditions, \eqref{eq:kkt_4}-\eqref{eq:kkt_6} are the primal feasibility conditions, and \eqref{eq:kkt_7}-\eqref{eq:kkt_8}  are the complementary slackness conditions. Note that \eqref{eq:kkt_1}-\eqref{eq:kkt_6} match the KKT conditions of \eqref{eq:ep_gen1}-\eqref{eq:ep_gen6} and that \eqref{eq:kkt_7}-\eqref{eq:kkt_8} match the KKT conditions of  \eqref{eq:ep_gen7}-\eqref{eq:ep_gen8}. Hence,  \eqref{eq:ccuc_chebyshev_socp}  and \eqref{eq:ep_gen} are characterized by the same set of KKT conditions and, thus, yield the equivalent solutions,  \cite{8246578,918286}. $\square$

\end{appendices}

    \bibliographystyle{IEEEtran}
    
    \bibliography{main}

\begin{thebibliography}{10}
\providecommand{\url}[1]{#1}
\csname url@samestyle\endcsname
\providecommand{\newblock}{\relax}
\providecommand{\bibinfo}[2]{#2}
\providecommand{\BIBentrySTDinterwordspacing}{\spaceskip=0pt\relax}
\providecommand{\BIBentryALTinterwordstretchfactor}{4}
\providecommand{\BIBentryALTinterwordspacing}{\spaceskip=\fontdimen2\font plus
\BIBentryALTinterwordstretchfactor\fontdimen3\font minus
  \fontdimen4\font\relax}
\providecommand{\BIBforeignlanguage}[2]{{%
\expandafter\ifx\csname l@#1\endcsname\relax
\typeout{** WARNING: IEEEtran.bst: No hyphenation pattern has been}%
\typeout{** loaded for the language `#1'. Using the pattern for}%
\typeout{** the default language instead.}%
\else
\language=\csname l@#1\endcsname
\fi
#2}}
\providecommand{\BIBdecl}{\relax}
\BIBdecl

\bibitem{pjm_value_of_the_market_2018}
\BIBentryALTinterwordspacing
{PJM Interconnection, LLC}. (2018) {The Value of Markets}. [Online]. Available:
  \url{https://tinyurl.com/y5ey3acs}
\BIBentrySTDinterwordspacing

\bibitem{CONEJO2018520}
A.~J. Conejo and R.~Sioshansi, ``Rethinking restructured electricity market
  design: Lessons learned and future needs,'' \emph{Int. J. El. Pwr. \& En.
  Syst.}, vol.~98, pp. 520 -- 530, 2018.

\bibitem{Bose2019}
S.~Bose and S.~H. Low, \emph{Some Emerging Challenges in Electricity
  Markets}.\hskip 1em plus 0.5em minus 0.4em\relax Cham: Springer International
  Publishing, 2019, pp. 29--45.

\bibitem{6656976}
J.~M. {Morales}, M.~{Zugno}, S.~{Pineda}, and P.~{Pinson}, ``Redefining the
  merit order of stochastic generation in forward markets,'' \emph{IEEE
  Transactions on Power Systems}, vol.~29, no.~2, pp. 992--993, March 2014.

\bibitem{osti_1251315}
\BIBentryALTinterwordspacing
{Weimar et al}, ``Integrating renewable generation into grid operations.''
  [Online]. Available: \url{https://www.osti.gov/biblio/1251315}
\BIBentrySTDinterwordspacing

\bibitem{5739566}
A.~Papavasiliou, S.~S. Oren, and R.~P. O'Neill, ``Reserve requirements for wind
  power integration: A scenario-based stochastic programming,'' \emph{IEEE
  Tran. Pwr. Syst.}, vol.~26, no.~4, pp. 2197--2206, Nov 2011.

\bibitem{doi:10.1287/opre.1090.0800}
G.~Pritchard, G.~Zakeri, and A.~Philpott, ``A single-settlement, energy-only
  electric power market for unpredictable and intermittent participants,''
  \emph{Operations Research}, vol.~58, no. 4-part-2, pp. 1210--1219, 2010.

\bibitem{6146389}
{J. M. Morales et al}, ``Pricing electricity in pools with wind producers,''
  \emph{IEEE Tran. Pwr. Syst.}, vol.~27, no.~3, pp. 1366--1376, Aug 2012.

\bibitem{4162624}
S.~Wong and J.~D. Fuller, ``Pricing energy and reserves using stochastic
  optimization in an alternative electricity market,'' \emph{IEEE Tran. Pwr.
  Syst.}, vol.~22, no.~2, pp. 631--638, May 2007.

\bibitem{8246578}
J.~{Kazempour}, P.~{Pinson}, and B.~F. {Hobbs}, ``A stochastic market design
  with revenue adequacy and cost recovery by scenario: Benefits and costs,''
  \emph{IEEE Tran. Pwr. Syst.}, vol.~33, no.~4, pp. 3531--3545, July 2018.

\bibitem{6153412}
C.~Ruiz, A.~Conejo, and S.Gabriel, ``Pricing nonconvexities in electricity
  pool,'' \emph{IEEE Tran. Pwr. Syst.}, vol.~27, no.~3, pp. 1334--42, Aug 2012.

\bibitem{7081775}
F.~Abbaspourtorbati and M.~Zima, ``The swiss reserve market: Stochastic
  programming in practice,'' \emph{IEEE Tran. Pwr. Syst.}, vol.~31, no.~2, pp.
  1188--1194, March 2016.

\bibitem{MORALES2014765}
{Juan M. Morales et al}, ``Electricity market clearing with improved scheduling
  of stochastic production,'' \emph{Eur. J. Op. Res.}, vol. 235, no.~3, pp. 765
  -- 774, 2014.

\bibitem{7332992}
M.~Lubin, Y.~Dvorkin, and S.~Backhaus, ``A robust approach to chance
  constrained optimal power flow with renewable generation,'' \emph{IEEE Tran.
  Pwr. Syst.}, vol.~31, no.~5, pp. 3840--3849, Sep. 2016.

\bibitem{doi:10.1137/130910312}
D.~Bienstock, M.~Chertkov, and S.~Harnett, ``Chance-constrained {OPF}:
  Risk-aware network control,'' \emph{SIAM Review}, vol.~56, no.~3, 2014.

\bibitem{6652224}
{L. Roald et al}, ``Analytical reformulation of security constrained optimal
  power flow with probabilistic constraints,'' in \emph{2013 IEEE Grenoble
  Conference}, June 2013, pp. 1--6.

\bibitem{7268773}
{Y. Dvorkin et al}, ``Uncertainty sets for wind power generation,'' \emph{IEEE
  Tran. Pwr. Syst.}, vol.~31, no.~4, pp. 3326--3327, July 2016.

\bibitem{8600344}
Y.~Dvorkin, M.~Lubin, and L.~Roald, ``Chance constraints for improving the
  security of {AC} opf,'' \emph{IEEE Tran. Pwr. Syst.}, pp. 1--1, 2019.

\bibitem{open_power}
\BIBentryALTinterwordspacing
\emph{Open Power System Data Project}, 2019. [Online]. Available:
  \url{https://open-power-system-data.org/}
\BIBentrySTDinterwordspacing

\bibitem{ONEILL2005269}
{R. O'Neill et al}, ``Efficient market-clearing prices in markets with
  nonconvexities,'' \emph{Eur. J. Op. Res.}, vol. 164, no.~1, pp. 269 -- 285,
  2005.

\bibitem{8329412}
{X. Kuang et al}, ``Pricing chance constraints in electricity markets,''
  \emph{IEEE Tran. Pwr. Syst.}, vol.~33, no.~4, pp. 4634--4636, July 2018.

\bibitem{6942382}
{Y. Dvorkin et al}, ``Assessing flexibility requirements in power systems,''
  \emph{{IET Gen., Tran. \& Dist.}}, vol.~8, no.~11, pp. 1820--1830, 2014.

\bibitem{7973099}
W.~Xie and S.~Ahmed, ``Distributionally robust chance constrained optimal power
  flow with renewables: A conic reformulation,'' \emph{IEEE Tran. Pwr. Syst.},
  vol.~33, no.~2, pp. 1860--1867, March 2018.

\bibitem{DVORKIN2019521}
\BIBentryALTinterwordspacing
V.~Dvorkin, J.~Kazempour, and P.~Pinson, ``Electricity market equilibrium under
  information asymmetry,'' \emph{Operations Research Letters}, vol.~47, no.~6,
  pp. 521 -- 526, 2019. [Online]. Available:
  \url{http://www.sciencedirect.com/science/article/pii/S0167637719300756}
\BIBentrySTDinterwordspacing

\bibitem{8017474}
L.~Roald and G.~Andersson, ``Chance-constrained {AC} optimal power flow:
  Reformulations and efficient algorithms,'' \emph{IEEE Tran. Pwr. Syst.},
  vol.~33, no.~3, pp. 2906--2918, May 2018.

\bibitem{6714594}
H.~{Wu}, M.~{Shahidehpour}, Z.~{Li}, and W.~{Tian}, ``Chance-constrained
  day-ahead scheduling in stochastic power system operation,'' \emph{IEEE
  Transactions on Power Systems}, vol.~29, no.~4, pp. 1583--1591, July 2014.

\bibitem{XieThesis}
\BIBentryALTinterwordspacing
W.~Xie, ``Relaxations and approximations of chance constrained stochastic
  programs,'' Ph.D. dissertation, Georgia Institute of Technology, 2017.
  [Online]. Available: \url{https://bit.ly/2VlFr5R}
\BIBentrySTDinterwordspacing

\bibitem{SUMMERS2015116}
{T. Summer et al}, ``Stochastic {OPF} based on conditional value at risk and
  robustness,'' \emph{Int. J. El. Pwr. \& En. Syst.}, vol.~72, pp. 116 -- 125,
  2015.

\bibitem{KUANG2019123}
{X. Kuang, A. Lamadrid and L. Zuluaga}, ``Pricing in non-convex markets with
  quadratic deliverability costs,'' \emph{{En. Econ.}}, vol.~80, pp. 123 --
  131, 2019.

\bibitem{SIOSHANSI2014143}
R.~Sioshansi, ``Pricing in centrally committed electricity markets,''
  \emph{Utilities Policy}, vol.~31, pp. 143 -- 145, 2014.

\bibitem{7039273}
D.~{Krishnamurthy}, W.~{Li}, and L.~{Tesfatsion}, ``An 8-zone test system based
  on iso new england data: Development and application,'' \emph{IEEE Tran. Pwr.
  Syst.}, vol.~31, no.~1, pp. 234--246, Jan 2016.

\bibitem{Li2017An8I}
\BIBentryALTinterwordspacing
W.~Li and L.~Tesfatsion, ``An 8-zone iso-ne test system with physically-based
  wind power,'' 2017. [Online]. Available: \url{https://bit.ly/2QvDdfc}
\BIBentrySTDinterwordspacing

\bibitem{8651383}
L.~{Exizidis}, J.~{Kazempour}, A.~{Papakonstantinou}, P.~{Pinson}, Z.~{De
  Grève}, and F.~{Vallée}, ``Incentive-compatibility in a two-stage
  stochastic electricity market with high wind power penetration,'' \emph{IEEE
  Transactions on Power Systems}, vol.~34, no.~4, pp. 2846--2858, July 2019.

\bibitem{918286}
B.~E. {Hobbs}, ``Linear complementarity models of nash-cournot competition in
  bilateral and poolco power markets,'' \emph{IEEE Tran. Pwr. Syst.}, vol.~16,
  no.~2, pp. 194--202, May 2001.

\end{thebibliography}

\end{document}